\documentclass[10pt, conference, letterpaper]{IEEEtran}

\usepackage{listings}

\usepackage{multicol}

\usepackage{blindtext}
\usepackage{amsfonts}
\usepackage{amsmath, amsthm, amssymb}
\usepackage{times}
\usepackage{url}
\usepackage{verbatim}
\usepackage{array}
\usepackage{mathrsfs}
\usepackage{enumitem}
\usepackage{textcomp}

\newcommand{\subparagraph}{}
\usepackage[compact]{titlesec}
\usepackage[ruled, lined, linesnumbered]{algorithm2e}
\SetAlFnt{\small}

\SetCommentSty{mycommfont}

\usepackage{amsfonts}
\usepackage{bbm}
\usepackage{soul}
\usepackage{color}
\usepackage{graphics}
\usepackage{graphicx,epsfig}
\usepackage{lipsum,graphicx,subcaption}
\graphicspath{{figures/}}
\usepackage{caption}
\usepackage{subcaption}

\usepackage{cleveref}
\usepackage{wrapfig}
\usepackage{url}
\usepackage{url}
\usepackage{amsmath,amssymb,epsfig,amsthm,psfrag,epsf, multirow,wrapfig, graphicx}
\usepackage[english]{babel}
\usepackage{epstopdf}
\usepackage{float}
\usepackage{color}
\usepackage[table,xcdraw]{xcolor}
\usepackage{amsmath}
\usepackage{amsfonts}
\usepackage{amssymb}
\usepackage{mathtools}
\DeclarePairedDelimiter{\ceil}{\lceil}{\rceil}
\usepackage{tabularx}
\usepackage{adjustbox}
\usepackage{physics}

\usepackage{amsfonts}
\usepackage{amssymb}
\usepackage{caption}
\usepackage{subcaption}
\usepackage{tabularx}
\usepackage{comment}
\usepackage[mathscr]{euscript}
\usepackage{amsbsy}
\usepackage{stmaryrd}
\usepackage{mathtools}
\usepackage[utf8]{inputenc}
\graphicspath{{figs/}}

\usepackage{booktabs}
\usepackage{graphicx,epsfig}
\usepackage{graphics}
\usepackage{multirow}
\usepackage{rotating}
\usepackage{caption}
\usepackage{subcaption}
\usepackage{amsmath}
\usepackage{amssymb}
\usepackage{makecell}

\DeclareMathOperator*{\argmax}{arg\,max}

\usepackage{mathtools}
\usepackage{comment}
\usepackage{multirow}

\usepackage{footnote}
\usepackage{tablefootnote}
\usepackage{scalerel}

\usepackage{tabularx}
\usepackage{mdframed}
\usepackage{lipsum}
\makeatletter
\newcommand*{\rom}[1]{\expandafter\@slowromancap\romannumeral #1@}
\makeatother
\usepackage{nccmath}

\usepackage{relsize}

\newtheorem{theorem}{Theorem}

\newtheorem{corollary}{Corollary}

\newtheorem{lemma}{Lemma}

\renewcommand\IEEEkeywordsname{Keywords}

\newcommand{\note}[1]{{\color{blue}{#1}}} 
\newcommand{\inred}[1]{{\color{red}{#1}}} 
\newcommand{\ingreen}[1]{{\color{green}{#1}}} 

\DeclareMathSymbol{\shortminus}{\mathbin}{AMSa}{"39}

\newcommand{\smallDelta}{{\text{$\scaleto{\Delta}{5pt}$}}}

\allowdisplaybreaks
\begin{document}
\def\eg{\mbox{\em e.g.}, }


\title{Unified Characterization and Precoding for Non-Stationary Channels}

\author{
\begin{tabular}[t]{c@{\extracolsep{8em}}c} 
Zhibin Zou, Maqsood Careem, Aveek Dutta & Ngwe Thawdar \\
Department of Electrical and Computer Engineering & US Air Force Research Laboratory \\ 
University at Albany SUNY, Albany, NY 12222 USA & Rome, NY, USA \\
\{{zzou2, mabdulcareem, adutta\}@albany.edu} & 
ngwe.thawdar@us.af.mil
\end{tabular}
}


    
\maketitle

\begin{abstract}
Modern wireless channels are increasingly dense and mobile making the channel highly non-stationary. The time-varying distribution and the existence of joint interference across multiple degrees of freedom (\eg users, antennas, frequency and symbols) in such channels render conventional precoding sub-optimal in practice, and have led to historically poor characterization of their statistics. The core of our work is the derivation of a high-order generalization of Mercer's Theorem to decompose the non-stationary channel into constituent fading sub-channels (2-D eigenfunctions) that are jointly orthogonal across its degrees of freedom. Consequently, transmitting these eigenfunctions with optimally derived coefficients eventually mitigates any interference across these dimensions and forms the foundation of the proposed joint spatio-temporal precoding. The precoded symbols directly reconstruct the data symbols at the receiver upon demodulation, thereby significantly reducing its computational burden, by alleviating the need for any complementary decoding. These eigenfunctions are paramount to extracting the second-order channel statistics, and therefore completely characterize the underlying channel. Theory and simulations show that such precoding leads to ${>}10^4{\times}$ BER improvement (at 20dB) over existing methods for non-stationary channels.
\end{abstract}

\section{Problem Statement}
\label{sec:intro}

Precoding at the transmitter is investigated in the literature and is relatively tractable when the wireless channel is stationary, by employing the gamut of linear algebraic and statistical tools to ensure interference-free communication \cite{Cho2010MIMObook,fatema2017massive,CostaDPC1983}. However, there are many instances, in modern and next Generation 
propagation environments such as mmWave, V2X, and massive-MIMO networks, where 
the channel is 
statistically non-stationary \cite{wang2018survey,huang2020general,mecklenbrauker2011vehicular} (the distribution is a function of time). 
This leads to sub-optimal and sometimes catastrophic performance even with state-of-the-art precoding \cite{AliNS0219} due to two factors: 
a) the time-dependence of the channel statistics,
and b) the existence of interference both jointly and independently across multiple dimensions (space (users/ antennas), frequency or time) in communication systems that leverage multiple degrees of freedom (\eg MU-MIMO, OFDM, OTFS \cite{OTFS_2018_Paper}).
This necessitates a unified characterization of the statistics of wireless channels that can also incorporate time-varying statistics, and novel precoding algorithms 
warrant flat-fading over the higher-dimensional interference profiles in non-stationary channels. 
%
Our solution to the above addresses a challenging open problem in the literature \cite{2006MatzOP}\footnote{This work is funded by the Air Force Research Laboratory Visiting Faculty Research Program (SA10032021050367), Rome, New York, USA.} : ``how to decompose non-stationary channels into independently fading sub-channels (along each degree of freedom) and how to precode using them", which is central to both characterizing channels and minimizing interference. 

\begin{figure}
    \centering
    \includegraphics[width=\linewidth]{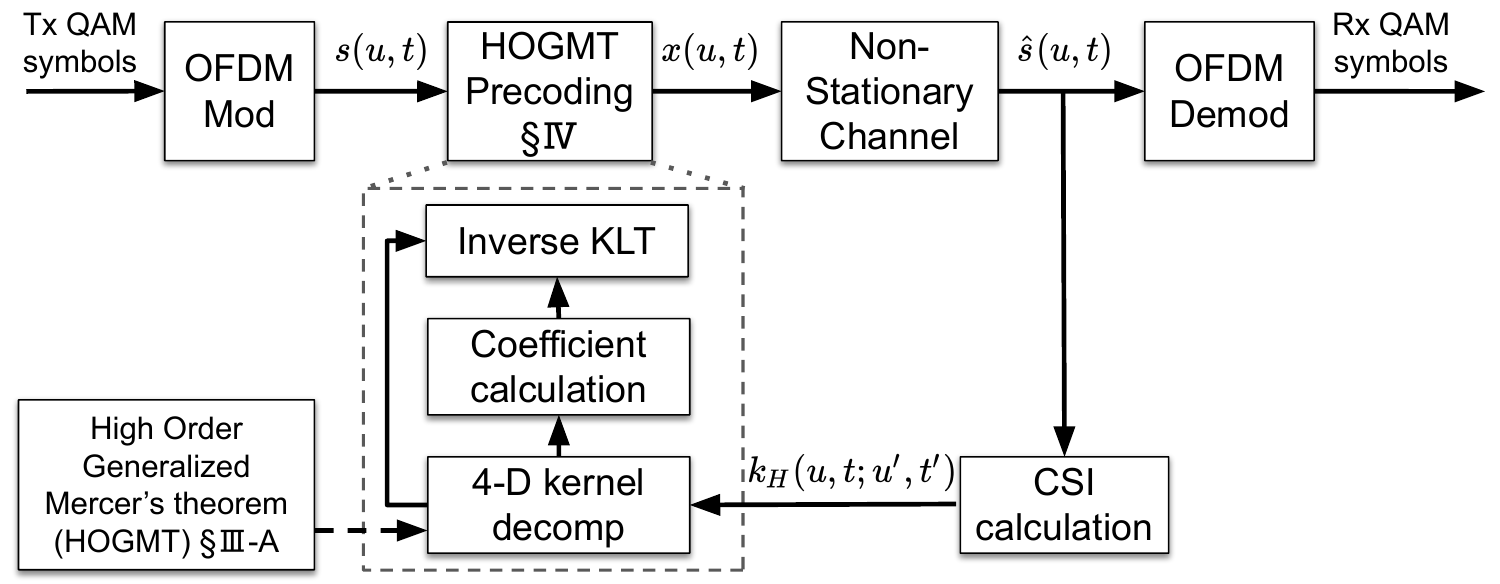}
    \caption{Unified characterization \& precoding
    } 
    \vspace{-20pt}
    \label{fig:Df}
\end{figure}
Unlike stationary channels,
the second order statistics of non-stationary channels are 4-dimensional, as they are functions of both time-frequency and delay-Doppler dimensions\footnote{Second order statistics of stationary channels 
depend only on the delay-Doppler (2-D) and hence can be extracted as a degenerate case of the non-stationary channel model, when its time-frequency dependence is constant.} \cite{MATZ20111}.
The core of our wireless channel characterization method, is the decomposition of this 4-dimensional channel kernel into 2-dimensional 
eigenfunctions that are jointly orthogonal across these dimensions, using a generalization of Mercer's Theorem to high-dimensional and asymmetric kernels. 
Unlike recent literature that only partially characterize the non-stationary channel using a select few local statistics~\cite{2018PMNS, bian2021general}, these eigenfunctions 
are used to extract any second-order statistics of the non-stationary channels that completely characterizes its distribution.  
Since any wireless channel model (\eg deterministic, stationary, frequency flat or selective) can be extracted from the general non-stationary channel kernel, the extracted eigenfunctions lead to a unified method to characterize the statistics of any wireless channel. 

Figure \ref{fig:Df} shows the data flow for joint spatio-temporal precoding at the transmitter.
The spatio-temporal CSI obtained from the receivers are used to construct a 4-dimensional channel kernel. 
In addition to spatial (inter-user or inter-antenna) or temporal (inter-symbol) interference, the time-varying kernel of non-stationary channels, induces joint
space-time interference.
We design a joint space-time precoding at the transmitter that 
involves combining the spatio-temporal eigenfunctions obtained by decomposing the space-time channel kernel, with optimal coefficients that 
minimize the least square error in the transmitted and received symbols.
Since the eigenfunctions are independently and jointly orthonomal sub-channels over space and time, 
precoding using them warrants flat-fading (interference-free communication) even in the presence of joint space-time interference. 
Further, these transmitted (precoded) symbols directly reconstruct the data symbols at the receiver when combined with calculated coefficients. 
Therefore, unlike existing precoding methods that require complementary decoding at the receiver\cite{Cho2010MIMObook}, we alleviate any need for complex receiver processing thereby significantly reducing its computational burden. 
Finally, the precoded symbols are scheduled to each user and are processed through the conventional transmitter signal processing blocks (\eg CP/ guard insertion) before transmission.
To the best of our knowledge, precoding for non-stationary channels is unprecedented in the literature. However, we include a comprehensive comparison with related precoding techniques in Appendix A-A.
\section{Background}
\label{sec:pre}

The wireless channel is typically expressed by a linear operator $H$, and the received signal $r(t)$ is given by $r(t){=}Hs(t)$, where $s(t)$ is the transmitted signal. The physics of the impact of $H$ on $s(t)$ is described using the delays and Doppler shifts in the multipath propagation~\cite{MATZ20111} given by \eqref{eq:H_delay_Doppler},
\begin{equation}
    r(t) = \sum\nolimits_{p=1}^P h_p s(t-\tau_p) e^{j2\pi \nu_p t}
    \label{eq:H_delay_Doppler}
\end{equation}
where $h_p$, $\tau_p$ and $\nu_p$ are the path attenuation factor, time delay and Doppler shift for path $p$, respectively. 
\eqref{eq:H_delay_Doppler} is expressed in terms of the overall delay $\tau$ and Doppler shift $\nu$ \cite{MATZ20111} in \eqref{eq:S_H}, 
%
\begin{align}
    r(t) &= \iint S_H(\tau, \nu) s(t{-} \tau) e^{j2\pi \nu t} ~d\tau ~d\nu \label{eq:S_H} \\
    & = \int L_H(t, f) S(f) e^{j2\pi tf} ~df \\
    & = \int h(t, \tau) s(t{-} \tau) ~d\tau \label{eq:h_relation}
\end{align}
where $S_H(\tau, \nu)$ is the \textit{(delay-Doppler) spreading function} of channel $H$, which describes the combined attenuation factor for all paths in the delay-Doppler domain. $S(f)$ is the Fourier transform of $s(t)$ and the time-frequency (TF) domain representation of $H$ is characterized by its \textit{TF transfer function}, $L_H(t,f)$, which is obtained by the 2-D Fourier transform of $S_H(\tau, \nu)$ as in \eqref{eq:TF_SH}. 
The time-varying impulse response $h(t,\tau)$ is obtained as the Inverse Fourier transform of $S_H(\tau, \nu)$ from the Doppler domain to the time domain as in  \eqref{eq:h_SH}.
\begin{align}
    L_H(t,f) &{=} \iint S_H(\tau, \nu) e^{j2\pi (t\nu{-} f \tau)} ~d\tau ~d\nu \label{eq:TF_SH}\\
    h(t,\tau) &{=} \int S_H(\tau, \nu) e^{j2\pi t \nu} ~d\nu \label{eq:h_SH}
\end{align}
%
%
Figures \ref{fig:SH} and \ref{fig:LH} show the time-varying response and TF transfer function for an example of a time-varying channel. 
For stationary channels, the TF transfer function is a stationary process with $\mathbb{E}\{L_H(t, f) L_H^*(t',f')\} {=} R_H(t{-}t',f{-}f')$, and the spreading function is a white process (uncorrelated scattering), i.e., $\mathbb{E}\{S_H(\tau, \nu) S_H^*(\tau',\nu')\} {=} C_H(\tau,\nu) \delta (\tau{-} \tau') \delta(\nu{-} \nu')$, where $\delta(\cdot)$ is the Dirac delta function.
$C_H(\tau,\nu)$ and $R_H(t-t',f-f')$ are the \textit{scattering function} and \textit{TF correlation function}, respectively, which are related via 2-D Fourier transform, 
\begin{equation}
    C_H(\tau,\nu) = \iint R_H(\Delta t, \Delta f) e^{-j2\pi(\nu \Delta t -  \tau \Delta f)} ~d\Delta t ~d\Delta f
\end{equation}
In contrast, for non-stationary channels, the TF transfer function is non-stationary process and the spreading function is a non-white process. 
Therefore, a \textit{local scattering function} (LSF) $\mathcal{C}_H(t,f;\tau,\nu)$ \cite{MATZ20111} is defined to extend $C_H(\tau,\nu)$ to the non-stationary channels in \eqref{eq:LSF}. 
Similarly, the \textit{channel correlation function} (CCF) $\mathcal{R}(\Delta t, \Delta f;\Delta \tau, \Delta \nu)$ generalizes $R_H(\Delta t, \Delta f)$ to the non-stationary case in \eqref{eq:CCF}.
\begin{align}
\begin{split}
    &\mathcal{C}_H(t,f ; \tau,\nu)  \label{eq:LSF}\\
    & {=} {\iint} R_L(t, f; \Delta t, \Delta f) e^{-j2\pi(\nu \Delta t -  \tau \Delta f)} ~d\Delta t ~d\Delta f  \\
    & {=} \iint R_S(\tau, \nu; \Delta \tau, \Delta \nu) e^{-j2\pi(t \Delta \nu -  f \Delta \tau)} ~d\Delta \tau ~d\Delta \nu \\
\end{split}\\
\begin{split}
    &\mathcal{R}(\Delta t, \Delta f;\Delta \tau, \Delta \nu) \label{eq:CCF}\\
    & {=} \iint R_L(t, f; \Delta t, \Delta f) e^{-j2\pi(\Delta \nu  t -  \Delta \tau  f)} ~d t ~d f \\
    & {=} \iint R_S(\tau, \nu; \Delta \tau, \Delta \nu) e^{-j2\pi(\Delta t \nu -  \Delta f  \tau)} ~d \tau ~d \nu
\end{split}
\end{align}
where $R_L(t, f; \Delta t, \Delta f){=}\mathbb{E}\{L_H(t, f {+} \Delta f) L_H^*(t {-} \Delta t, f)\}$ and $R_S(\tau, \nu; \Delta \tau, \Delta \nu){=}\mathbb{E}\{S_H(\tau, \nu {+} \Delta \nu) S_H^*(\tau{-} \Delta \tau, \nu)\}$. 
For stationary channels, CCF reduces to TF correlation function $\mathcal{R}(\Delta t, \Delta f;\Delta \tau, \Delta \nu){=}R_H(\Delta t, \Delta f) \delta(\Delta t) \delta(\Delta f)$.
%
\vspace{-10pt}
\begin{figure}[h]
\begin{subfigure}{.24\textwidth}
  \centering
\includegraphics[width=1\linewidth]{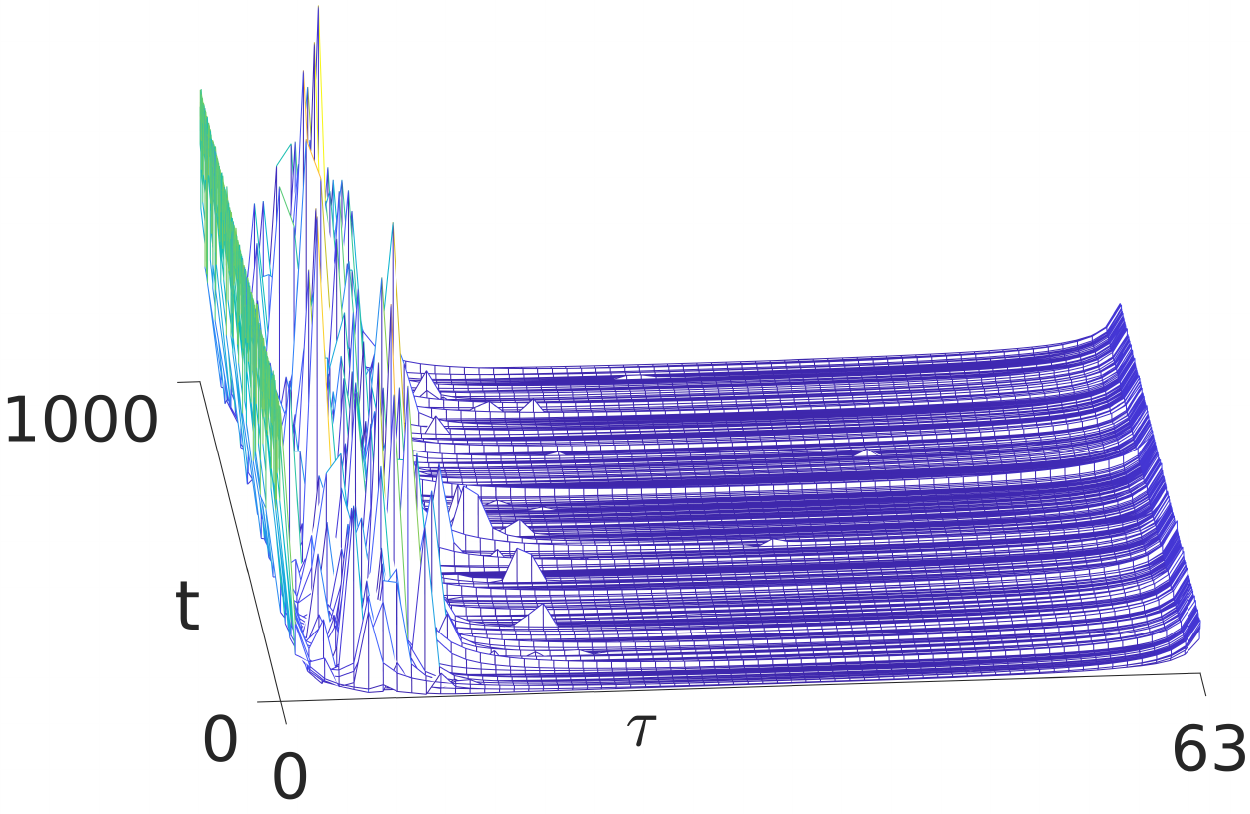}
  \vspace{-20pt}
  \caption{Time-varying response $h(t,\tau)$}
  \label{fig:SH}
\end{subfigure}
\begin{subfigure}{.24\textwidth}
  \centering
  \includegraphics[width=1\linewidth]{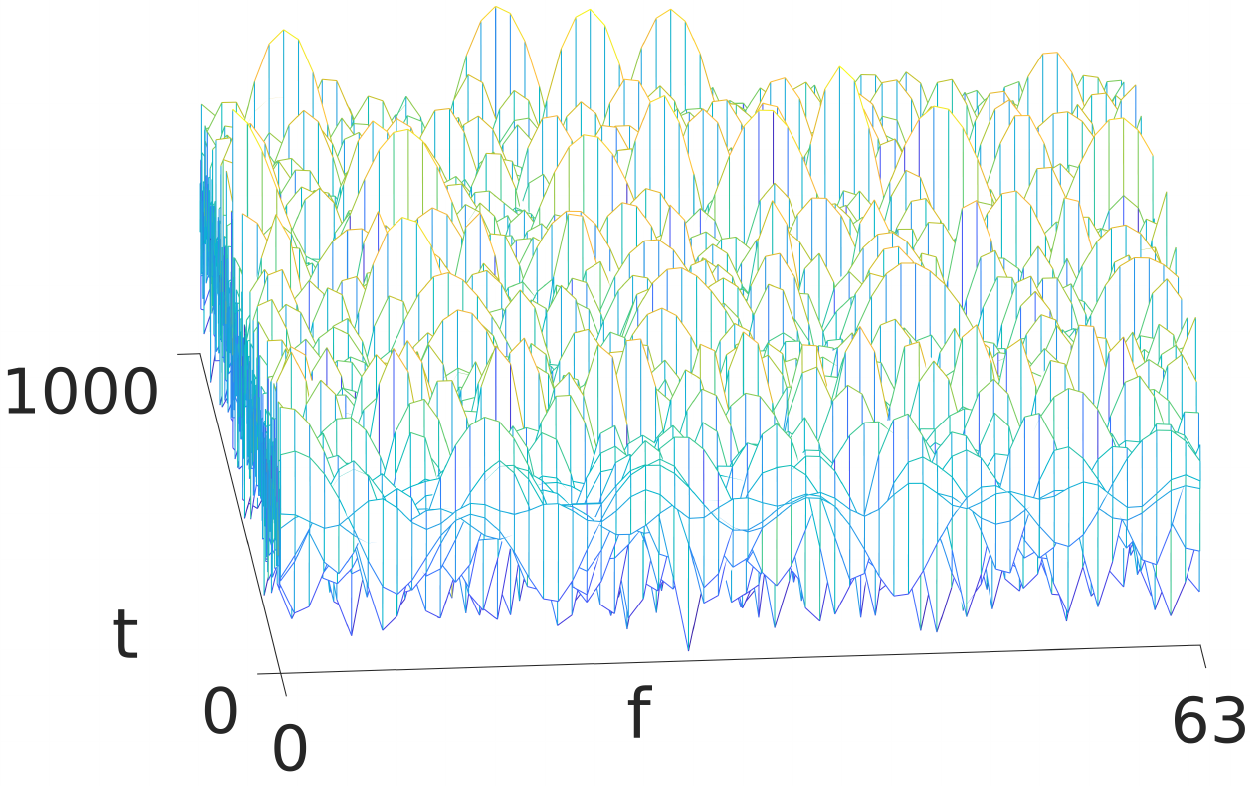}
  \vspace{-20pt}
  \caption{TF transfer function $L_H(t,f)$}
  \label{fig:LH}
\end{subfigure}
\caption{Illustration of a non-stationary channel}
  \vspace{-10pt}
\label{fig:toychannel}
\end{figure}

\section{Non-stationary channel decomposition and characterization}
\label{sec:NSdec}
The analysis of non-stationary channels is 
complicated as its statistics vary across both time-frequency and delay-Doppler domains resulting in 4-D second order statistics \cite{Matz2005NS}, which motivates the need for a unified characterization of wireless channels\footnote{Any channel can be generated as a special case of the non-stationary channel. Therefore a characterization of non-stationary channels would generalize to any other wireless channel \cite{MATZ20111}.}.
Wireless channels are completely characterized by their statistics, however they are difficult to extract for non-stationary channels, due to their time dependence.
%
%
Therefore, we start by expressing the channel $H$ using an atomic channel $G$ and the 4-D channel kernel $\mathcal{H}(t, f ; \tau, \nu)$~\cite{Matz2005NS} as in \eqref{eq:atomic}, 
\begin{align}
\label{eq:atomic}
   H = \iiiint \mathcal{H}(t, f ; \tau, \nu) G_{t,f}^{\tau,\nu} ~dt ~df ~d\tau ~dv
\end{align}
where $G$ is a normalized $(||G||{=}1)$ linear prototype system whose transfer function $L_G(t, f)$ is smooth and localized about the origin of the TF plane. $G_{t,f}^{\tau,\nu} {=} S_{t, f{+}\nu} G S_{t{-}\tau, f}^+$ means that the atomic channel $G$ shifts the signal components localized at $(t{-}\tau, f)$ to $(t, f{+}\nu)$ on the TF plane. $S_{\tau, \nu}$ is TF shift operator defined as $(S_{\tau, \nu}s)(t) {=} s(t{-}\tau) e^{j2 \pi \nu t}$. 
Then the channel kernel $\mathcal{H}(t, f ; \tau, \nu)$ is given by \eqref{eq:compute_H}.
\begin{align}
\label{eq:compute_H}
&\mathcal{H}(t, f ; \tau, \nu)=\left\langle H, G_{t, f}^{\tau, \nu}\right\rangle  \\
& {=}\mathrm{e}^{j 2 \pi f \tau} \iint L_{H}\left(t^{\prime}, f^{\prime}\right) L_{G}^{*}\left(t^{\prime}{\shortminus}t, f^{\prime}{\shortminus}f\right) 
 \mathrm{e}^{{-}j 2 \pi\left(\nu t^{\prime}{-}\tau f^{\prime}\right)} \mathrm{d} t^{\prime} \mathrm{d} f^{\prime} \nonumber
 \end{align}

The statistics of any wireless channel 
can always be obtained from the above 4-D channel kernel. Therefore, decomposing this kernel into fundamental basis allows us to derive a unified form to characterize any wireless channel.

\subsection{Channel decomposition}
4-D channel kernel decomposition into orthonormal 2-D kernels is unprecedented the literature, but is essential to mitigate joint interference in the 2-D space and to completely characterize non-stationary channels.  
While Mercer's theorem \cite{1909Mercer} provides a method to decompose symmetric 2-D kernels into the same eigenfunctions along different dimensions, it cannot directly decompose 4-D channel kernels due to their high-dimensionality and since the kernel is not necessarily symmetric in the time-frequency delay-Doppler domains. 
Karhunen–Loève transform (KLT) \cite{wang2008karhunen}
provides a method to decompose kernels into component eigenfunctions of the same dimension, however is unable to decompose into orthonormal 2-D space-time eigenfunctions, and hence cannot be used to mitigate interference on the joint space-time dimensions. 
Therefore, we derive an 
asymmetric 4-dimensional kernel decomposition method that 
combines the following two steps as shown in figure \ref{fig:thm}:
A) A generalization of Mercer's theorem that is applicable to both symmetric or asymmetric kernels, and B) An extension of KLT for high-dimensional kernels. 



\begin{lemma}
\label{lemma:GMT}
(Generalized Mercer's theorem (GMT))
The decomposition of a 2-D process $K {\in} L^2(X {\times} Y)$, where $X$ and $Y$ are square-integrable zero-mean 
processes, 
is given by,
\begin{align}
    K(t, t') = \sum\nolimits_{n=1}^{\infty} \sigma_n \psi_n(t) \phi_n(t')
    \label{eq:GMT}
\end{align}
where $\sigma_n$ is a random variable with $\mathbb{E}\{\sigma_n \sigma_{n'}\} {=} \lambda_n \delta_{nn'}$, and $\lambda_n$ is the $n$\textsuperscript{th} eigenvalue. $\psi_n(t)$ and $\phi_n(t')$ are eigenfunctions.
\end{lemma}
\noindent
The proof combines 
Mercer's Theorem with KLT to generalize it to asymmetric kernels and is provided in Appendix B-A in the external document \cite{appendix_link}. 
%
From Lemma~\ref{lemma:GMT}, by letting $\rho(t, t'){=}\psi_n(t) \phi_n(t')$ in \eqref{eq:GMT} we have \eqref{eq:GMT_2D}, 
\begin{equation}
    K(t, t') {=} \sum_{n=1}^{\infty} \sigma_n \rho_n(t, t')
    \label{eq:GMT_2D}
\end{equation}
where the 2-D kernel is decomposed into random variable ${\sigma_n}$ with constituent 
2-D eigenfunctions, $\rho(t, t')$, this serves as an extension of KLT 
to 2-D kernels.
A similar extension leads to the derivation of KLT for N-dimensional kernels which is 
key to deriving Theorem \ref{thm:hogmt}.



\noindent
\fbox{\begin{minipage}{0.97\linewidth}
\begin{theorem}
\label{thm:hogmt}
(High Order GMT (HOGMT)) The decomposition of $M {=} Q{+}P$ dimensional kernel $K {\in} L^M(X {\times} Y)$, where $X(\gamma_1,\cdots,\gamma_Q)$ and $Y(\zeta_1,\cdots,\zeta_P)$ are $Q$ and $P$ dimensional kernels respectively, that are square-integrable zero-mean random processes, is given by \eqref{eq:col},
\begin{align}
\label{eq:col}
K(\zeta_1{,}{...}{,}\zeta_P{;} \gamma_1{,}{...}{,} \gamma_Q) {=} {\sum_{n{=}1}^ \infty} \sigma_n \psi_n(\zeta_1{,}{...}{,}\zeta_P) \phi_n(\gamma_1{,}{...}{,}\gamma_Q)
\end{align}
where $\mathbb{E}\{\sigma_n \sigma_n'\} {=} \lambda_n \delta_{nn'}$. $\lambda_n$ is the $n$\textsuperscript{th} eigenvalue and $\psi_n(\zeta_1,\cdots,\zeta_P)$ and $\phi_n(\gamma_1,\cdots, \gamma_Q)$ are $P$ and $Q$ dimensional eigenfunctions respectively.
\end{theorem}
\end{minipage}}

The proof is provided in Appendix~B-B in \cite{appendix_link}. 
Theorem~\ref{thm:hogmt} is applicable to any $M$ dimensional channel kernel. Examples of such channel kernels may include 1-D time-varying channels, 2-D time-frequency kernels for doubly dispersive channels \cite{2003MatzDP}, user, antenna dimensions in MU-MIMO channels and angles of arrivals and departures 
in mmWave channels.
\begin{figure}[t]
    \centering
    \includegraphics[width=\linewidth]{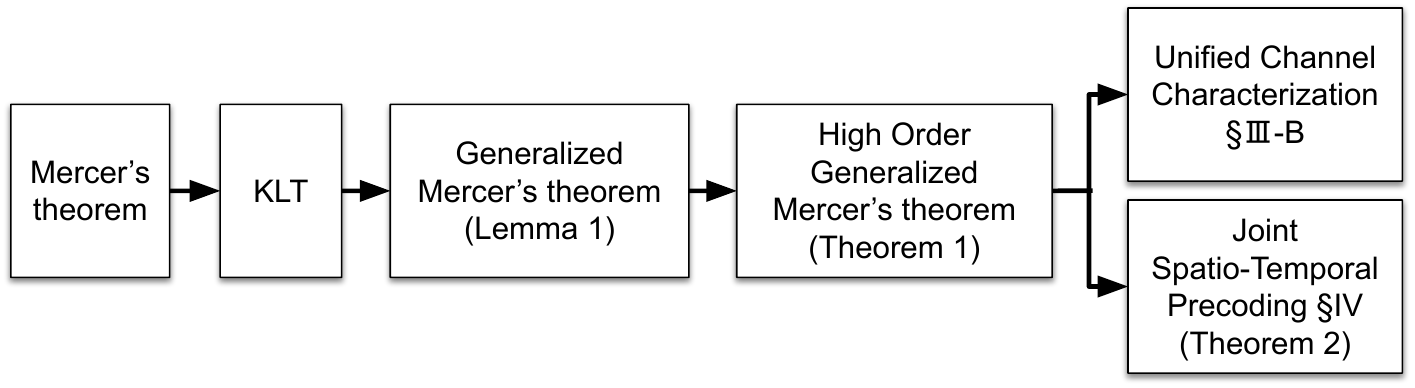}
    \vspace{-10pt}
    \caption{Derivation of High Order Generalized Mercer's Theorem for channel decomposition}
    \vspace{-10pt}
    \label{fig:thm}
\end{figure}
Theorem~\ref{thm:hogmt} ensures that the 4-D channel kernel in \eqref{eq:compute_H} is decomposed as in \eqref{eq:channel_kernel_decomp} into 2-D eigenfunctions that are jointly orthonormal in the time-frequency dimensions as in \eqref{eq:properties}.
\begin{align}
\label{eq:channel_kernel_decomp}
&\mathcal{H}(t, f ; \tau, \nu) = \sum\nolimits_{n{=1}}^\infty \sigma_n \psi_n(t, f) \phi_n(\tau, \nu)\\
\label{eq:properties}
&\begin{aligned}
&& \iint \psi_n(t, f) \psi_{n'}^*(t, f) ~dt ~df {=} \delta_{nn'}  \\
&& \iint \phi_n(\tau, \nu) \phi_{n'}^*(\tau, \nu) ~d\tau ~d\nu {=} \delta_{nn'}
\end{aligned}
\end{align}
Moreover, from \eqref{eq:channel_kernel_decomp} and \eqref{eq:properties} we have that, 
\begin{align}
    \iint \mathcal{H}(t{,}f{;}\tau{,}\nu) \phi_n^*(\tau{,}\nu)~d\tau ~d\nu {=} \sigma_n \psi_n(t{,}f)
    \label{eq:duality}
\end{align}
\eqref{eq:duality} suggests that 
when the eigenfunction, $\phi_n^*(\tau,\nu)$ is transmitted through the channel a different eigenfunction, $\psi_n(t,f)$ is received with $\sigma_n$. 
Therefore, we refer to $\phi_n$ and $\psi_n$ as a pair of \textit{dual} eigenfunctions.
By definition, these eigenfunctions are constituent sub-channels of the channel kernel that only undergo a scaling when transmitted over channel.
Therefore, the dual eigenfunctions are referred to as \textit{flat-fading sub-channels} of $H$.

\subsection{Unified Characterization Using Channel Statistics}
\label{sec:stat}

Wireless channels are fully characterized by their (second order) statistics, which we calculate using the extracted eigenvalues and 2-D eigenfunctions.
The CCF is calculated as the correlations of $\mathcal{H}(t, f ; \tau, \nu)$~\cite{Matz2005NS} and is given by,
\begin{align}
    & \abs{ \mathcal{R}(\Delta t, \Delta f;\Delta \tau, \Delta \nu) } \label{eq:CCF3}\\
    &{=}  \Bigl| {\iiiint} \mathbb{E}\{\mathcal{H}^*(t {\shortminus} {\smallDelta} t{,} f {\shortminus} {\smallDelta} f{;} \tau {\shortminus} {\smallDelta} \tau{,} \nu {\shortminus} {\smallDelta} \nu) 
    \mathcal{H}(t{,}f{;}\tau{,} \nu)\} d t d f d \tau d \nu \Bigr|  \nonumber \\
    &{=} \sum\nolimits_{n{=1}}^\infty \lambda_n \abs{R_{\psi_n}(\Delta t,\Delta f)}  |R_{\phi_n}(\Delta \tau,\Delta \nu)| \label{eq:CCF2}
\end{align}
where \eqref{eq:CCF2} is obtained by substituting \eqref{eq:duality} in \eqref{eq:CCF3}. 
$R_{\psi_n}(\Delta t,\Delta f)$ and $R_{\phi_n}(\Delta \tau,\Delta \nu)$ are the correlations of $\psi_n(t, f)$ and $\phi_n(\tau, \nu)$, respectively. 
The LSF reveals the non-stationarities (in time or frequency) in a wireless channel and is given by the 4-D Fourier transform ($\mathbb{F}^{4}$) of the CCF as, 
%
%
\begin{align}
\label{eq:LSF_eigen}
    & \mathcal{C}_H(t, f; \tau, \nu) {=} \mathbb{F}^{4}\left\{\mathcal{R}({\smallDelta} t{,} {\smallDelta} f{;}{\smallDelta} \tau{,} {\smallDelta} \nu)\right\} \nonumber\\
    & {=} \iiiint \mathcal{R}({\smallDelta} t{,} {\smallDelta} f{;}{\smallDelta} \tau{,} {\smallDelta} \nu)
    \mathrm{e}^{{\shortminus}j 2 \pi(t \Delta \nu{\shortminus}f \Delta \tau{+}\tau \Delta f{\shortminus}\nu \Delta t)} \mathrm{d} t \mathrm{d} f \mathrm{d} \tau \mathrm{d} \nu \nonumber \\ 
    & {=} \sum\nolimits_{n{=1}}^\infty \lambda_n |\psi_n( \tau,\nu)|^2 |\phi_n(t, f)|^2
\end{align}
where 
$|\psi_n( \tau,\nu)|^2$ and $|\phi_n(t, f)|^2$ represent the spectral density of $\psi_n(t,f)$ and $\phi_n(\tau,\nu)$, respectively.
%
Then, the 
\textit{global (or average) scattering function} $\overline{C}_H(\tau, \nu)$ and (local) TF path gain $\rho_H^2 (t,f)$ \cite{Matz2005NS} are calculated in \eqref{eq:GSF} and \eqref{eq:rho},
%
\begin{align}
\label{eq:GSF}
    & \overline{C}_H(\tau, \nu) {=} \mathbb{E}\{|S_H(\tau, \nu)|^2\} = \iint \mathcal{C}_H(t, f; \tau, \nu) ~dt ~df  \\
    & \rho_H^2 (t,f) {=}  \mathbb{E}\{|L_H(t, f)|^2\} = \iint \mathcal{C}_H(t, f; \tau, \nu) ~d\tau ~dv 
\label{eq:rho}
\end{align}
\eqref{eq:GSF} and \eqref{eq:rho} are re-expressed in terms of the spectral density of eigenfunctions by using \eqref{eq:LSF_eigen} and the properties in \eqref{eq:properties},  

\begin{align}
    &\overline{C}_H(\tau, \nu) {=} \mathbb{E}\{|S_H(\tau, \nu)|^2\} {=} \sum\nolimits_{n{=1}}^\infty \lambda_n |\psi_n( \tau,\nu)|^2 \\
    &  \rho_H^2 (t,f) {=} \mathbb{E}\{|L_H(t, f)|^2\} {=} \sum\nolimits_{n{=1}}^\infty \lambda_n |\phi_n(t, f)|^2 
\end{align}

Finally, the \textit{total transmission gain} $\mathcal{E}_H^2$ is obtained by integrating the LSF out with respect to all four variables, 
\begin{align}
\label{TFgain}
    & \mathcal{E}_H^2 = \iiiint \mathcal{C}_H(t, f; \tau, \nu) ~dt ~df ~d \tau ~d\nu = \sum\nolimits_{n{=}1}^ \infty \lambda_n
\end{align}
Therefore, the statistics of the non-stationary channel is completely characterized by its eigenvalues and eigenfunctions 
obtained by the decomposition of $\mathcal{H}(t,f;\tau,\nu)$, which are summarized in Table~\ref{tab:cha}.
\setlength{\textfloatsep}{0.1cm}
\setlength{\tabcolsep}{0.2em}
\begin{table}[h]
\caption{Unified characterization of non-stationary channel}
\renewcommand*{\arraystretch}{1.2}
\centering
\begin{tabular}{|l|l|}
\hline
Statistics & Eigen Characterization \\ \hline
CCF $|\mathcal{R}(\Delta t, \Delta f;\Delta \tau, \Delta \nu)|$             &   ${\sum} \lambda_n |R_{\psi_n}(\Delta t{,}\Delta f)| |R_{\phi_n}(\Delta \tau{,}\Delta \nu)|$                     \\ \hline
LSF  $\mathcal{C}_H(t, f; \tau, \nu)$             &         ${\sum} \lambda_n |\psi_n( \tau,\nu)|^2 |\phi_n(t, f)|^2$               \\ \hline
Global scattering function   $\overline{C}_H(\tau{,} \nu)$               &        ${\sum} \lambda_n |\psi_n( \tau,\nu)|^2 $                \\ \hline
Local TF path gain $\rho_H^2 (t,f)$
    & ${\sum} \lambda_n |\phi_n( t, f)|^2 $  \\ \hline
Total transmission gain     $\mathcal{E}_H^2$                          &     ${\sum} \lambda_n$                   \\ \hline
\end{tabular}
\label{tab:cha}
\vspace{-10pt}
\end{table}
%
\section{Joint Spatio-Temporal Precoding} 
\label{sec:Precoding}

The kernel $\mathcal{H}(t, f ; \tau, \nu)$ in \eqref{eq:compute_H} describes the time-frequency delay-Doppler response of the channel from \eqref{eq:h_relation} and is essential to extract the statistics of the non-stationary channel $H$ for a single user as they depend on the same 4 dimensions.
For precoding, we express the spatio-temporal channel response 
by extending the time-varying response $h(t,\tau)$ to incorporate multiple users, i.e., $h_{u,u'}(t,\tau)$ \cite{2017GhazalNSmodel,almers2007survey}, which denotes the time-varying impulse response between the $u'$\textsuperscript{th} transmit antenna and the $u$\textsuperscript{th} user, where each user has a single antenna. The 4-D spatio-temporal channel estimation is widely investigated \cite{2020CESrivastava,2005CEXiaoli,2008CEMilojevic} . In this paper we assume perfect CSI at transmit. Thus the received signal in \eqref{eq:h_relation} is extended as
\begin{align}
    \label{eq:MU}
    r_u (t) &= \int \sum\nolimits_{u'} h_{u,u'}(t,\tau) s_{u'}(t-\tau) d\tau + v_{u}(t) \nonumber \\
    & = \int \sum\nolimits_{u'}  k_{u,u'}(t,t') s_{u'}(t') dt' + v_{u}(t)
\end{align}
where $v_u(t)$ is the noise, $s_u(t)$ is the data signal and $k_{u,u'} (t, t') {=} h_{u,u'}(t, t{-}t')$ is the channel kernel.
Then, the relationship between the transmitted and received signals is obtained by rewriting \eqref{eq:MU} in its continuous form in \eqref{eq:MU2}.
\begin{align}
    \label{eq:MU2}
    & r(u,t) = \iint k_H(u,t;u',t') s(u',t')~du'~dt' + v(u,t)
\end{align}

Let $x(u,t)$ be the precoded signal, then the corresponding received signal is $Hx(u,t)$. The aim of precoding in this work is to minimize the interference, i.e., to minimize the least square error, $\|s(u,t) {-} Hx(u,t)\|^2$.

%

\begin{lemma}
\label{lemma:bestprojection}
Given a non-stationary channel $H$ with kernel $k_H(u{,}t{;}u'{,}t')$, if each projection in $\{H \varphi_n(u{,}t)\}$ are orthogonal to each other, there exists a precoded signal 
scheme 
$x(u,t)$ that ensures 
interference-free communication at the receiver, 
\begin{align}
    & \|s(u,t) - Hx(u,t)\|^2 = 0 
\label{eq:obj}
\end{align}
where $\varphi_n(u,t)$ is the eigenfunction of $x(u,t)$, obtained by KLT decomposition as $x(u,t) {=} \sum_{n=1}^\infty x_n \varphi_n(u,t)$.
\vspace{-5pt}
\end{lemma}
The proof is provided in Appendix C-A in \cite{appendix_link}.
%
Therefore, precoding using $\{\varphi_n\}{=}\{\phi_n\}$ or $\{\varphi_n\}{=}\{\psi_n\}$) obtained by decomposing the 
channel kernel using Theorem \eqref{thm:hogmt}, (i.e., constructing $x(u,t)$ using $\{\phi_n\}$ or $\{\psi_n\}$ with coefficients $x_n$ using inverse KLT, eventually leads to interference-free communication, as it satisfies \eqref{eq:obj} by ensuring that $\{H \varphi_n(u,t)\}$ are orthogonal using the properties in \eqref{eq:properties}.

\noindent
\fbox{\begin{minipage}{0.97\linewidth}
\begin{theorem}
\label{thm:thm2}
(HOGMT-based precoding) Given a non-stationary channel $H$ with kernel $k_H(u,t;u',t')$, the precoded signal $x(u,t)$ that ensures interference-free communication at the receiver is constructed by inverse KLT as,
\begin{align}
\label{eq:x}
    x(u{,}t) {=} \sum_{n{=}1}^\infty x_n \phi_n^*(u{,}t),\text{where},
    x_n {=} \frac{{\langle} s(u{,}t){,} \psi_n(u{,}t) {\rangle}}{\sigma_n}
\end{align}
where 
$\{ \sigma_n \}$, $\{ \psi_n\}$ and $\{ \phi_n\} $ are obtained by decomposing the kernel $k_H(u,t;u',t')$ using Theorem~\ref{thm:hogmt} as in \eqref{eq:thm2_decomp},
\begin{align}
\label{eq:thm2_decomp}
&k_H(u,t;u',t') = \sum\nolimits_{n{=1}}^\infty  \sigma_n \psi_n(u,t) \phi_n(u',t')
\end{align}
\end{theorem}
\end{minipage}}
The proof is provided in Appendix C-B. 
Although precoding involves a linear combination of $\phi_n^*(u,t)$ with $x_n$ it is a non-linear function ($\mathcal{W}(\cdot)$) with respect to the data signal $s(u,t)$, i.e., $x(u,t) {=} \mathcal{W}(\{\phi_n(u,t)\} ; \{\psi_n(u,t)\}, \{\sigma_n\}, s(u,t))$. \eqref{eq:properties} and \eqref{eq:thm2_decomp} suggest that the 4-D kernel is decomposed into jointly orthogonal sub-channels $\{\psi_n(u,t)\}$ and $\{\phi_n(u',t')\}$.
Therefore, the precoding in Theorem \ref{thm:thm2} can be explained as transmitting the eigenfunctions $\{\phi_n^*(u,t)\}$ after multiplying with specific coefficients $\{x_n\}$.
Consequently, when transmitted through the channel $H$, it transforms $\{\phi_n(u,t)\}$ to its dual eigenfunctions $\{\psi_n(u,t)\}$ with $\{\sigma_n\}$ as in \eqref{eq:them2_duality},
\begin{align}
    \iint k_H(u,t;u',t') \phi_n^*(u, t) ~du ~dt {=} \sigma_n \psi_n(u', t').
    \label{eq:them2_duality}
\end{align}
which means that, $H \phi_n^*(u,t){=}\sigma_n \psi_n(u,t)$, as proved in Appendix C-B in \cite{appendix_link}.
%
Then the data signal $s(u,t)$ is directly reconstructed at the receiver (to the extent of noise $v_u(t)$) as the net effect of precoding and propagation in the channel ensures that from \eqref{eq:MU}, $r(u,t){=}Hx(u,t){+}v_u(t){\rightarrow}s(u,t){+}v_u(t){=}\hat{s}(u,t)$ using Lemma \ref{lemma:bestprojection}, 
where $\hat{s}(u,t)$ is the estimated signal. 
Therefore, the spatio-temporal decomposition of the channel in Theorem \ref{thm:hogmt} allows us to precode the signal such that all interference in the spacial domain, time domain and joint space-time domain are cancelled when transmitted through the channel, leading to a joint spatio-temporal precoding scheme. 
Further, this precoding ensures that the data signal is reconstructed directly at the receiver with an estimation error that 
of $v_u(t)$, 
thereby completely pre-compensating the spatio-temporal fading/ interference in non-stationary channels to the level of AWGN noise.
Therefore, this precoding does not require complementary decoding at receiver, which vastly reducing its hardware and computational complexity compared to state-of-the-art precoding methods like Dirty Paper Coding (DPC) or linear precoding (that require a complementary decoder \cite{vu2006exploiting}).

\begin{figure}[t]
\begin{subfigure}{.225\textwidth}
  \centering
  \includegraphics[width=1\linewidth]{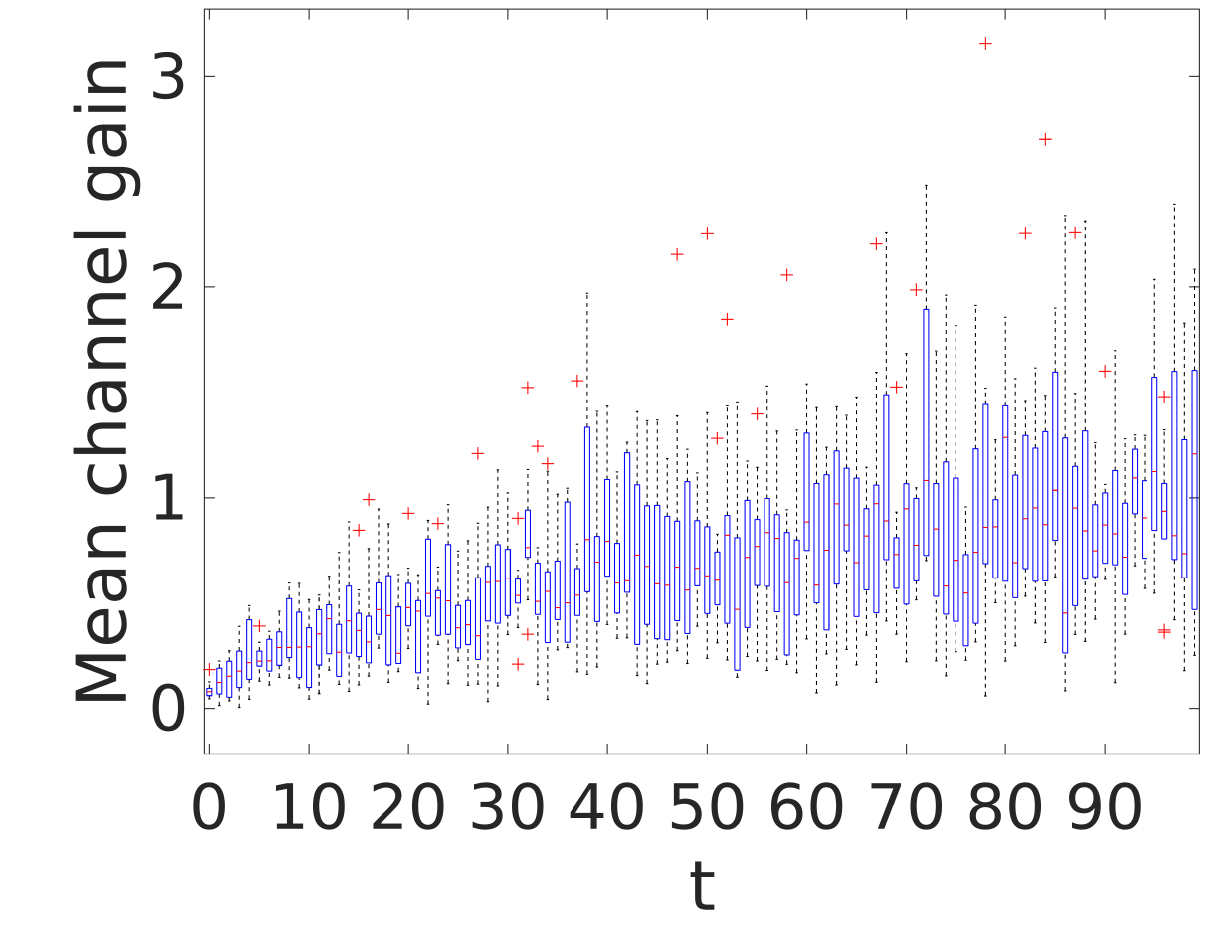}
  \caption{Distribution of mean channel gains at each time instance. 
  } 
  \label{fig:stat}
\end{subfigure}
\begin{subfigure}{.26\textwidth}
  \centering
  \includegraphics[width=1\linewidth]{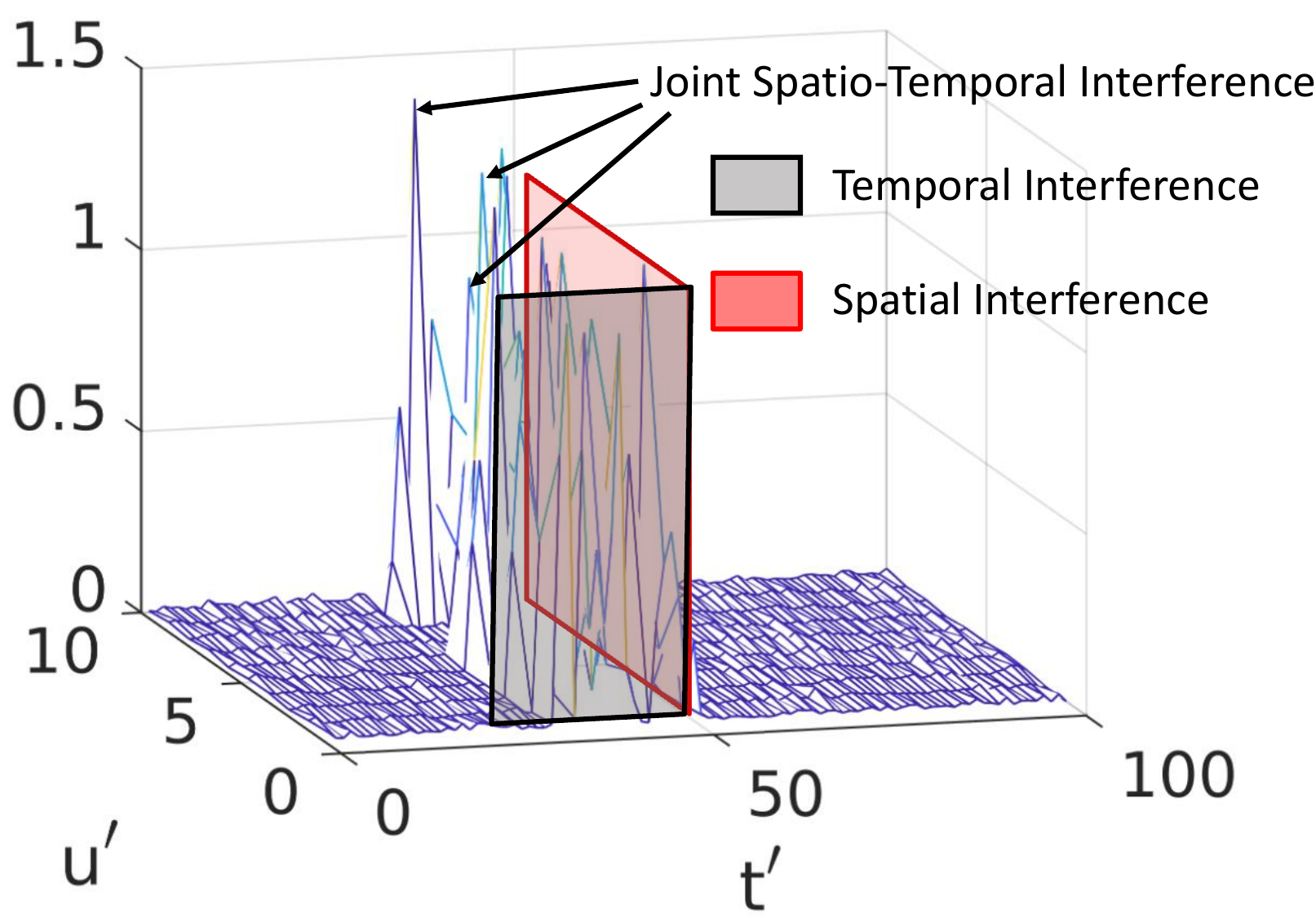}
  \caption{Channel kernel $k_H(u,t;u',t')$ for $u {=} 1$ and $t {=} 50$}
  \label{fig:joint_interference_kernel}
\end{subfigure}
  \caption{Non-stationary channel statistics and kernel}
  \label{fig:ns_channel}
\end{figure}
Theorem~\ref{thm:thm2} does not make any assumptions on the type, dimensions or size of the channel kernel. 
Corollary \ref{col:EP_space} demonstrates the application of Theorem \ref{thm:thm2} to an example of a deterministic multi-user channel where only spatial interference from other users' exist. The received signal is given by $r_u{=}\sum_{u'} h_{u,u'} s_{u'} {+} v_n$ and its continuous form is given by, 
\begin{align}
\label{eq:h_space}
    r(u){=} {\int} k_H (u{,}u') s(u')d u' + v(u)
\end{align}
\begin{corollary}
\label{col:EP_space}
Given a deterministic multi-user channel kernel $k(u,u')$, the precoded signal $x(u)$ that warrants spatial interference-free reception is given by \eqref{eq:xut}. 
\begin{equation}
\label{eq:xut}
    x(u) = \sum_{n{=}1}^\infty \frac{\langle s(u), \psi_n(u) \rangle}{\sigma_n} \phi_n^*(u) 
\end{equation}
where $\{ \sigma_n \}$, $\{ \psi_n\}$ and $\{ \phi_n\} $ follow from Theorem~\ref{thm:hogmt} for the 2-D case, i.e., $k_H(u;u'){=}\sum_{n{=}1}^\infty  \sigma_n \psi_n(u) \phi_n(u')$.
\end{corollary}
Proof of Corollary \ref{col:EP_space} is provided in Appendix C-C of \cite{appendix_link}.
The precoding in Corollary \ref{col:EP_space} holds for any other 2-D channel kernels like the stationary/non-stationary single-user channel kernel given by $k_H(t,t')$ by replacing $k_H(u,u')$ with $k_H(t,t')$.
\section{Results}

We analyze the accuracy of the proposed unified channel characterization and joint spatio-temporal precoding using a non-stationary channel simulation framework in Matlab.
The simulation environment considers $10$ mobile receivers (users) and $100$ time instances of a non-stationary 4-D kernel $k_H(u,t;u',t')$, where the number of delayed symbols (delay taps) causing interference are uniformly distributed between $[10,20]$ symbols for each user at each time instance. 

The pre-processing of the 4-D channel kernel and 2-D data symbols involves mapping them to a low-dimensional space using an invertible mapping 
$f{:}u{\times}t{\to}m$.
Although Theorem \ref{thm:hogmt} decomposes the channel kernel into infinite 
eigenfunctions, we show that it is sufficient to decompose the channel kernel into a finite number of eigenfunctions and select only those whose eigenvalues are greater than a threshold value $\epsilon^2$ for precoding, i.e., $\{(\phi_n(\cdot),\psi_n(\cdot)): \sigma_n{>}\epsilon\}$.
These eigenfunctions are used to calculate the coefficients for joint spatio temporal precoding, which subsequently construct the precoded signal after inverse KLT and combining the real and imaginary parts. 


\begin{figure}[t]
\begin{subfigure}{.24\textwidth}
  \centering
  \includegraphics[width=1\linewidth]{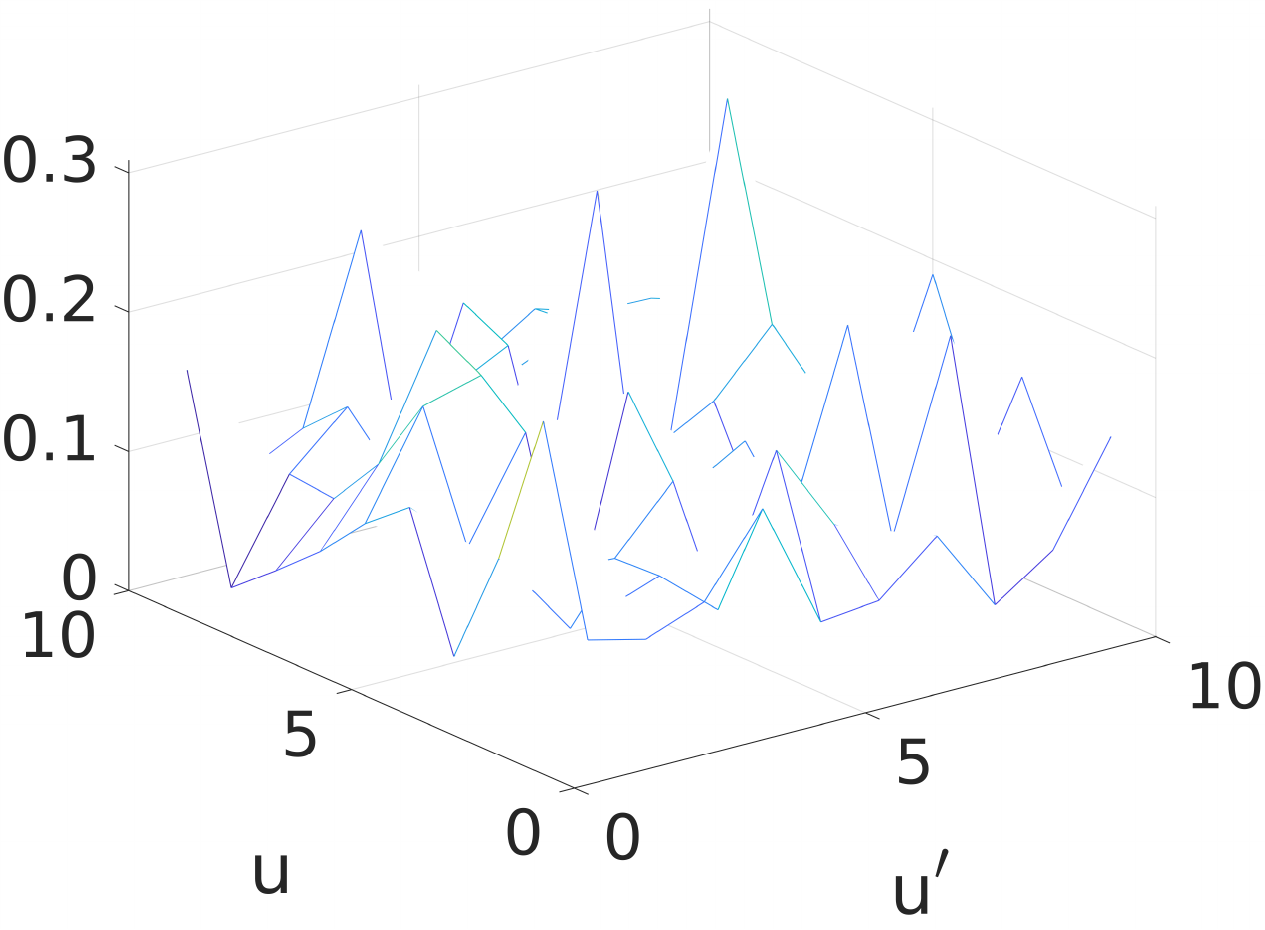}
  \caption{4-D kernel $k_H(u,t;u',t')$ at $t {=} 1$ and $t' {=} 1$ which shows the spatial interference}
  \label{fig:hs1}
\end{subfigure}
\begin{subfigure}{.24\textwidth}
  \centering
  \includegraphics[width=1\linewidth]{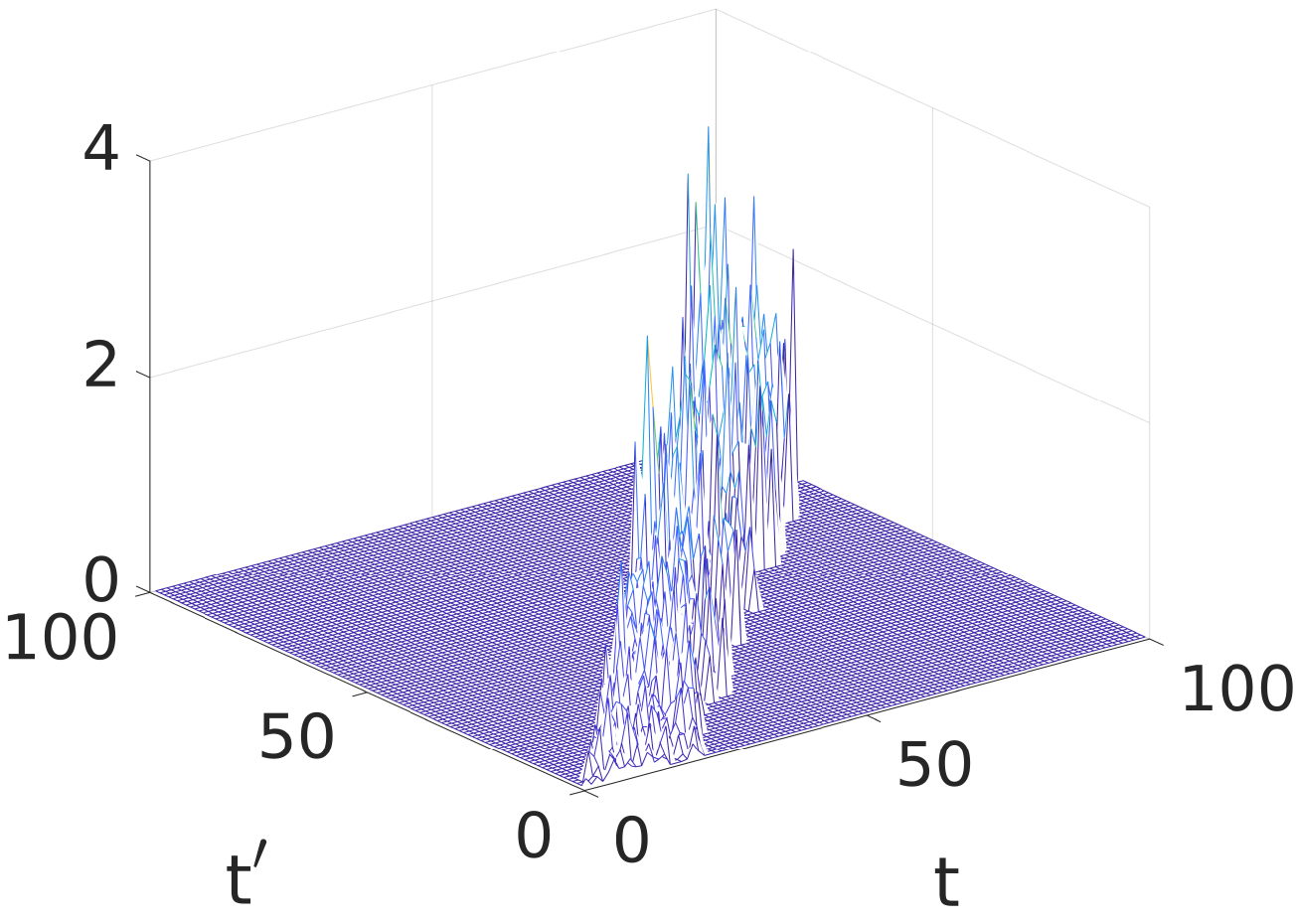}
  \caption{4-D kernel $k_H(u,t;u',t')$ at $u {=} 1$ and $u' {=} 1$ which shows the temporal interference}
  \label{fig:ht1}
\end{subfigure}
  \caption{Separately spatial and delay effects of 4-D kernel}
\end{figure}

%
%
%
%
Figure~\ref{fig:stat} shows distribution of the statistics of the channel gain (mean and variance) for each time instance of the non-stationary channel and further corroborates its non-stationarity. 
%
%
%
%

\begin{figure*}[t]
\centering
\begin{subfigure}{.27\textwidth}
  \centering
  \includegraphics[width=1\linewidth]{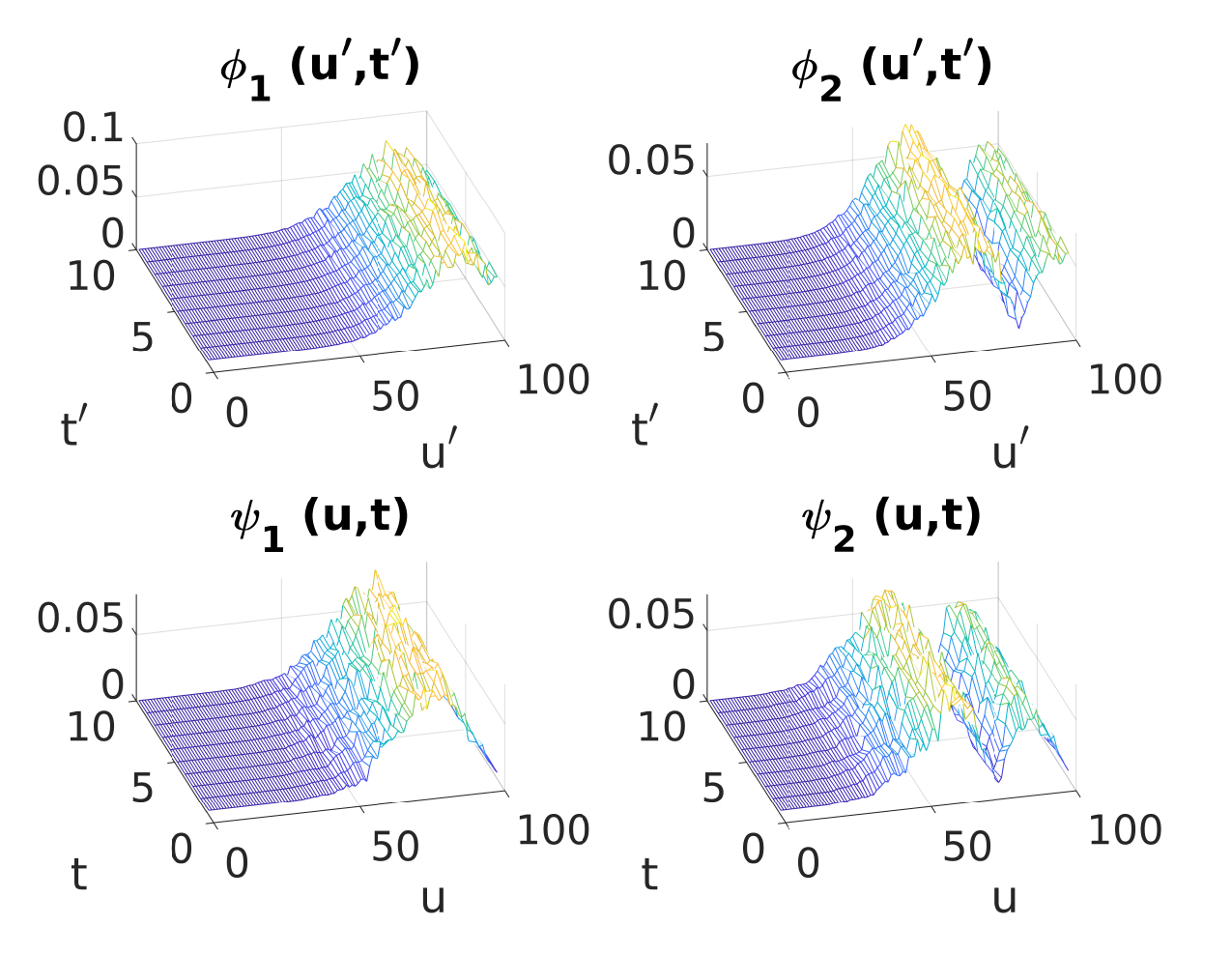}
  \caption{Dual spatio-temporal eigenfunctions decomposed from kernel $k_H(u,t;u',t')$ by HOGMT}
  \label{fig:Eigenfunctions}
\end{subfigure}
\qquad
\begin{subfigure}{.27\textwidth}
  \centering
  \includegraphics[width=1\linewidth]{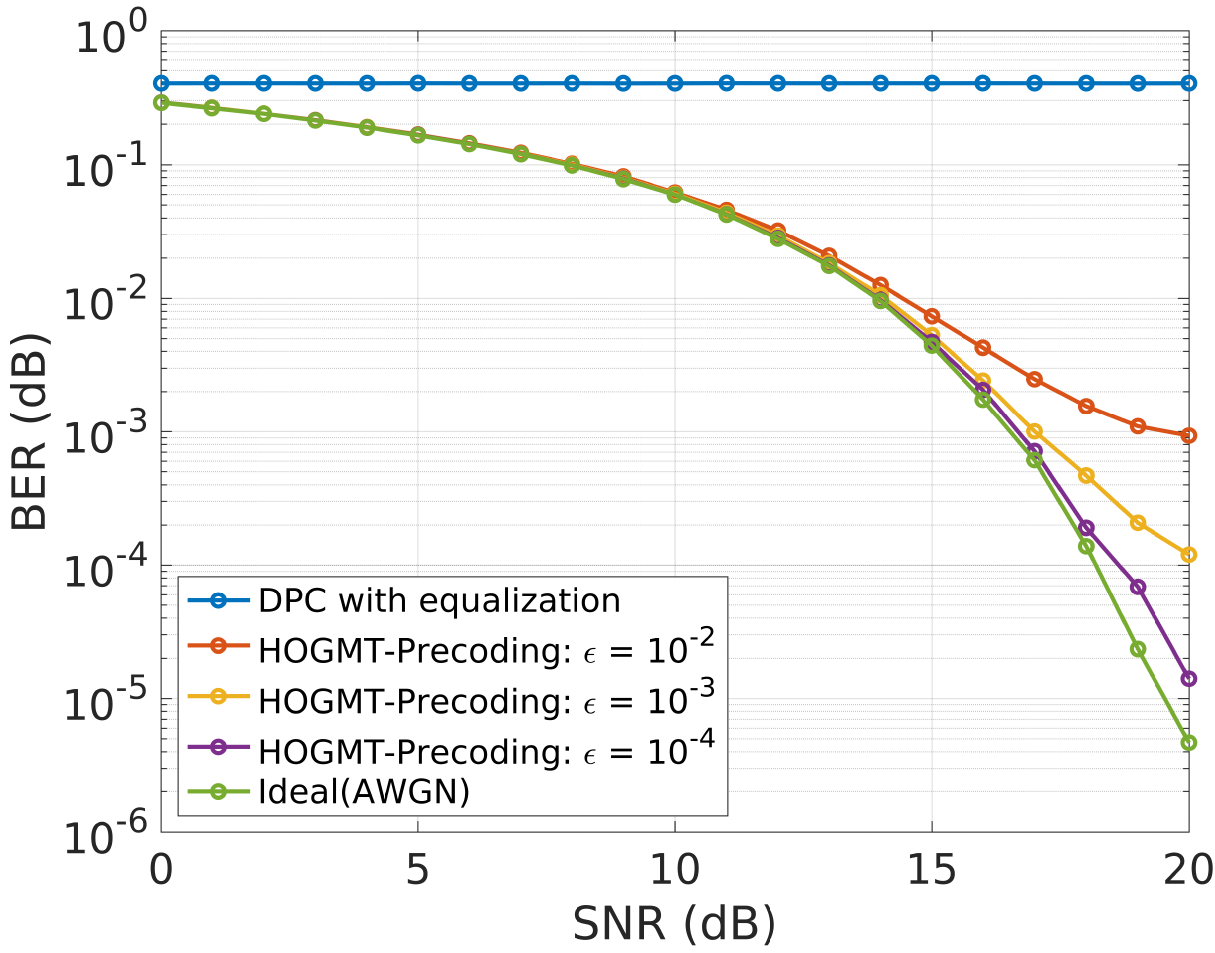}
  \caption{BER of HOGMT based spatial-temporal precoding for different $\epsilon$ and comparison with the state-of-the-art}
  \label{fig:ber_space_time}
\end{subfigure}
\qquad
\begin{subfigure}{.27\textwidth}
  \centering
  \includegraphics[width=1\linewidth]{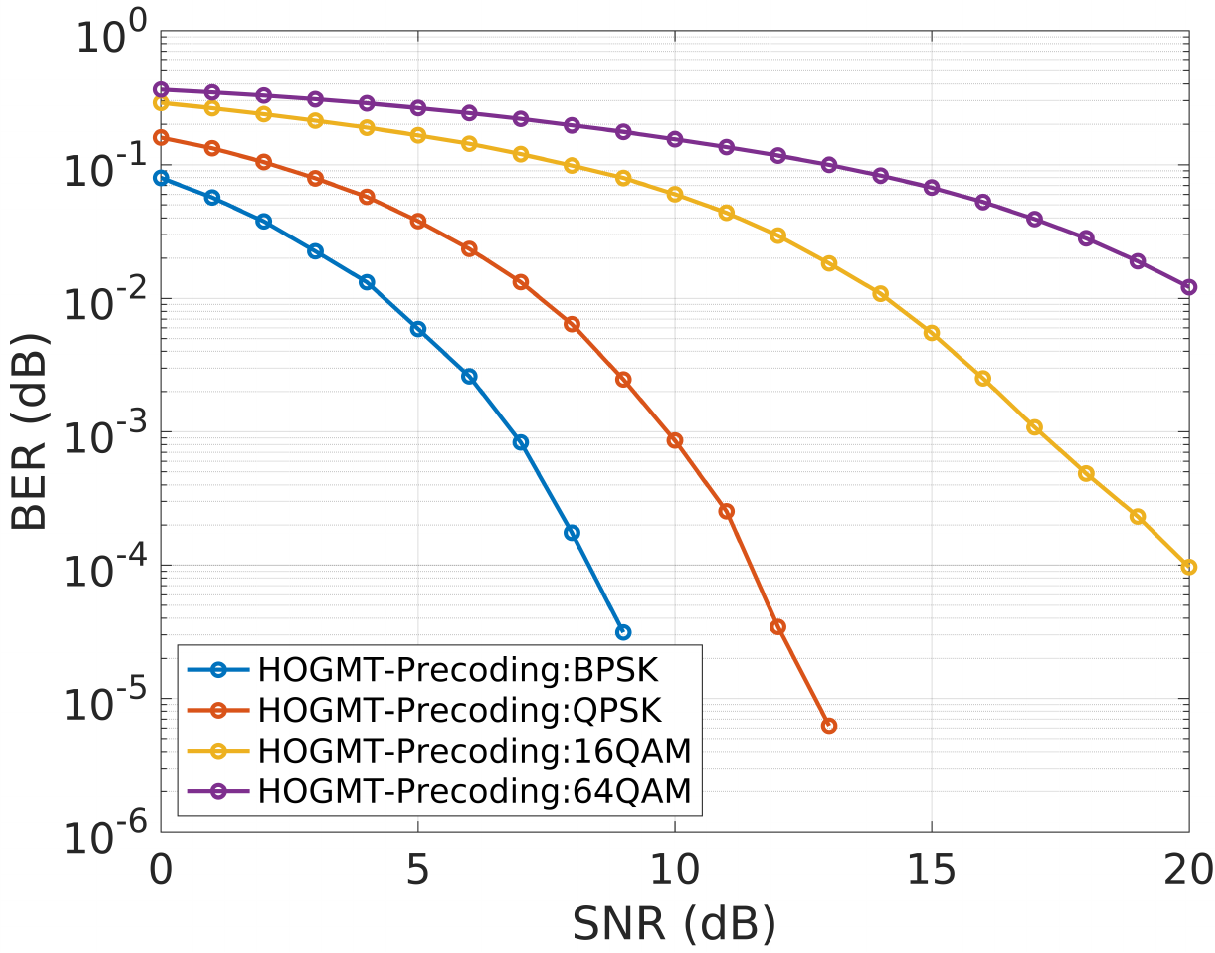}
  \caption{BER of HOGMT based spatio-temporal precoding for BPSK, QPSK ,16-QAM and 64-QAM modulations}
  \label{fig:Ber_space_time_multimod}
\end{subfigure}
  \caption{HOGMT based spatio-temporal precoding}
  \vspace{-20pt}
\end{figure*}
Figure~\ref{fig:joint_interference_kernel} shows the channel response for user $u {=} 1$ at $t{=}50$ which indicates the manifestation of interference. For $u{=}1$, the spatial interference from other users (inter-user interference) occurs on the red plane (at $t'{=50}$), while the temporal interference at $t{=}50$ occurs due to previous delayed symbols (inter-symbol interference) on the grey plane (at $u'{=}1$). However, the received symbols at user $u{=}1$ is also affected by other delayed symbols from other users (i.e., $t'{<50}$ for $u'{\neq}1$), which leads to joint spatio-temporal interference, and necessitates joint precoding over space-time dimensions. 
The varitation of the channel kernels over time and for different users is demonstrated and further explained in Appendix D for completion.
This is the cause of joint space-time interference which necessitates joint precoding in the 2-dimensional space using eigenfunctions that are jointly orthogonal.

Figure~\ref{fig:hs1} shows the spatial (inter-user) interference caused by other users for a fixed time instance (time instances $t{=}1$ and $t'{=}1$) in terms of the 4-D channel kernel, i.e., $k_H(u,1,u',1)$. 
Figure~\ref{fig:ht1} shows the temporal (inter-symbol) interference for a user $u{=}1$ caused by its own (i.e., $u'{=}1$) delayed symbols (\eg due to multipath), i.e., $k_H(1,t;1,t')$. We also observe that the temporal interference for each user occurs from it's own 20 immediately delayed symbols. 
Figure \ref{fig:Eigenfunctions} shows two pairs of dual spatio-temporal eigenfunctions $(\phi_n(u',t'),\psi_n(u,t))$ (absolute values) 
obtained by decomposing $k_H(u,t;u',t')$ in \eqref{eq:thm2_decomp}.
We see that this decomposition is indeed asymmetric as each $\phi_n(u',t')$ and $\psi_n(u',t')$ are not equivalent (but shifted), and that each $\phi_1(u',t')$ and $\psi_1(u,t)$ are jointly orthogonal with $\phi_2(u',t')$ and $\psi_2(u,t)$ as in \eqref{eq:properties}, respectively.
Therefore, when $\phi_1(u',t')$ (or $\phi_2(u',t')$) is transmitted through the channel, the dual eigenfunctions, $\psi_1(u,t)$ (or $\psi_2(u,t)$) is received with $\sigma_1$ and $\sigma_2$, respectively.
Therefore, the non-stationary 4-D channel is decomposed to dual flat-fading sub-channels.
%
%

Figures \ref{fig:ber_space_time} shows the 
BER at the receiver, using joint spatio-temporal precoding (HOGMT-precoding) at the transmitter with 16-QAM modulated symbols for non-stationary channels.
Since this precoding is able to cancel all interference that occurs in space, time and across space-time dimensions which are shown in figure \ref{fig:joint_interference_kernel}, 
it achieves significantly lower BER over existing precoding methods that employ DPC 
at the transmitter to cancel spatial interference and Zero Forcing equalization at the receiver to mitigate temporal interference. 
Further, we show that with sufficient eigenfunctions ($\epsilon{=}10^{-4}$ in this case), proposed method can achieve near ideal BER, only 0.5dB more SNR to achieve the same BER as the ideal case, where the ideal case assumes all interference is cancelled and only AWGN noise remains at the receiver.
This gap exists since practical implementation employs a finite number of eigenfunctions as opposed to an infinite number in \eqref{eq:thm2_decomp}. 
Figure \ref{fig:Ber_space_time_multimod} compares the BER of HOGMT based spatio-temporal precoding for various modulations (BPSK, QPSK, 16-QAM and 64-QAM) for the same non-stationary channel with $\epsilon{=}10^{-3}$.
As expected we observe that the lower the order of the modulation, the lower the BER but at the cost of lower data rate. 
However, we observe that even with high-order modulations (\eg 64-QAM) the proposed precoding achieves low BER (${\approx}10^{-2}$ at SNR${=}20$dB), allowing high data-rates even over challenging non-stationary channels. The choice of the order of the modulation is therefore, based on the desired BER and data rate for different non-stationary scenarios.



Figure \ref{fig:spatial_precoding} shows an example of precoding for an example of a multi-user deterministic channel defined by the 2-D kernel, $K_H(u,u')$ with 30 mobile users as defined in \eqref{eq:h_space}, where only spatial interference from other users exists. 
The spatial interference from other users is portrayed in figure \eqref{fig:Hs} in terms of $K_H(u,u')$. The achieved BER at the receiver by precoding using Corollary \ref{col:EP_space} is shown in figure \ref{fig:ber_space} and is compared with the state-of-the-art DPC for spatial precoding \cite{CostaDPC1983}. 
The performance of HOGMT based spatial precoding for various modulations with $\epsilon{=}10^{-3}$ is compared in figure \ref{fig:Ber_space_multimode}.
\begin{figure*}[t]
\centering
\begin{subfigure}{.27\textwidth}
  \centering
  \includegraphics[width=1\linewidth]{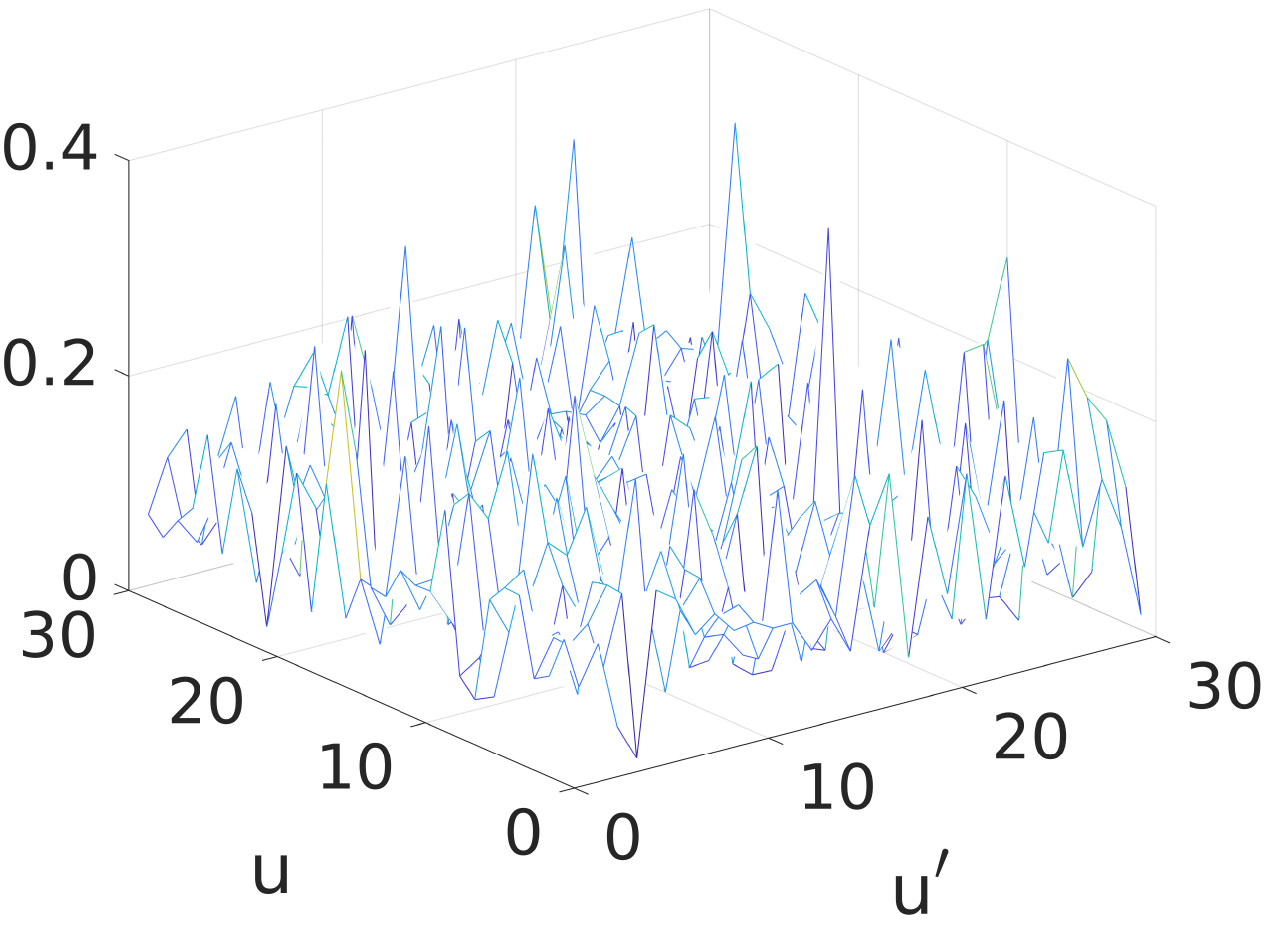}
  \caption{2-D kernel $k_H(u,u')$ where only spatial interference exists}
  \label{fig:Hs}
\end{subfigure}
\qquad
\begin{subfigure}{.27\textwidth}
  \centering
  \includegraphics[width=1\linewidth]{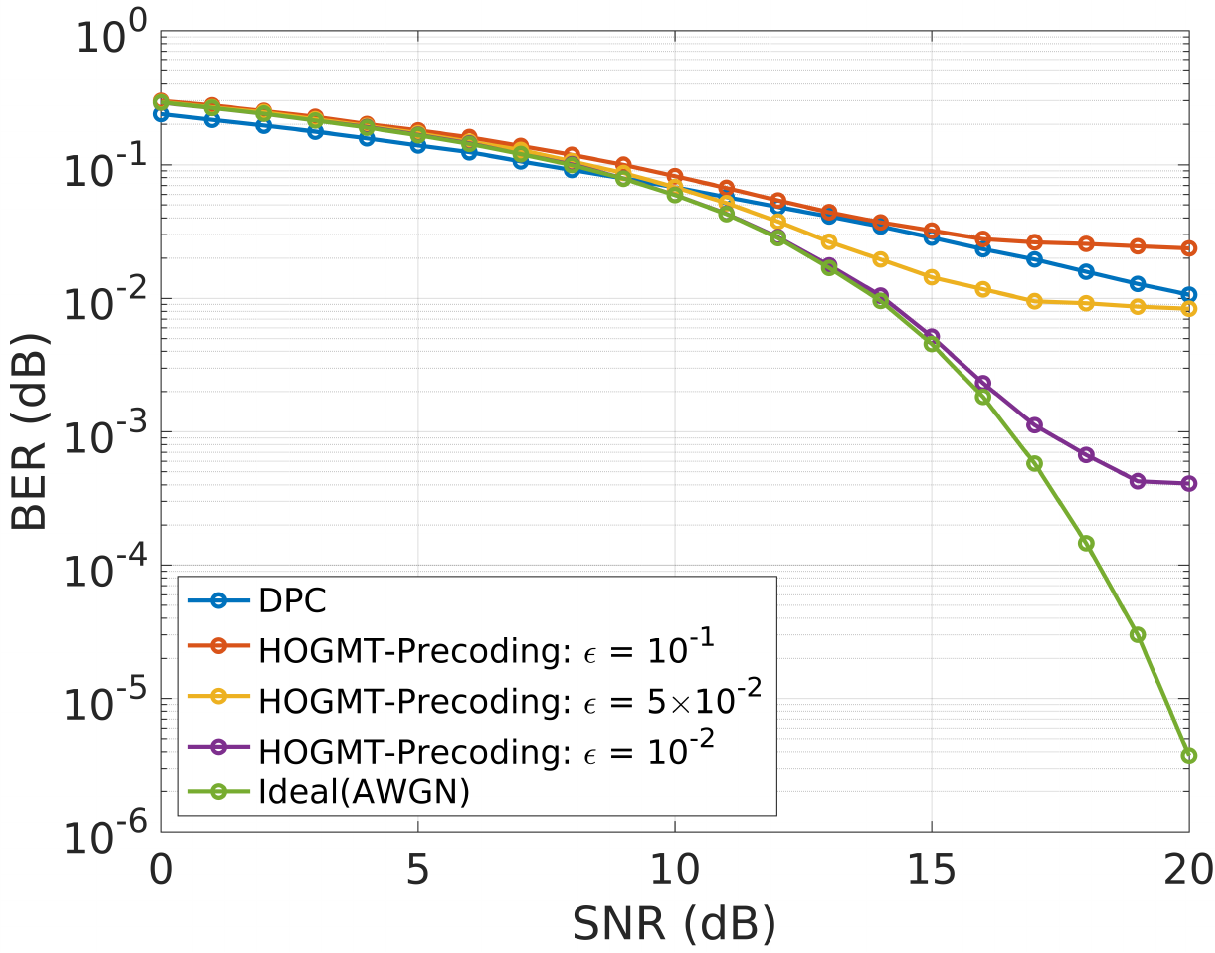}
  \caption{BER of HOGMT based spatial precoding for different $\epsilon$ and comparison with the state-of-the-art}
  \label{fig:ber_space}
\end{subfigure}
\qquad
\begin{subfigure}{.27\textwidth}
  \centering
  \includegraphics[width=1\linewidth]{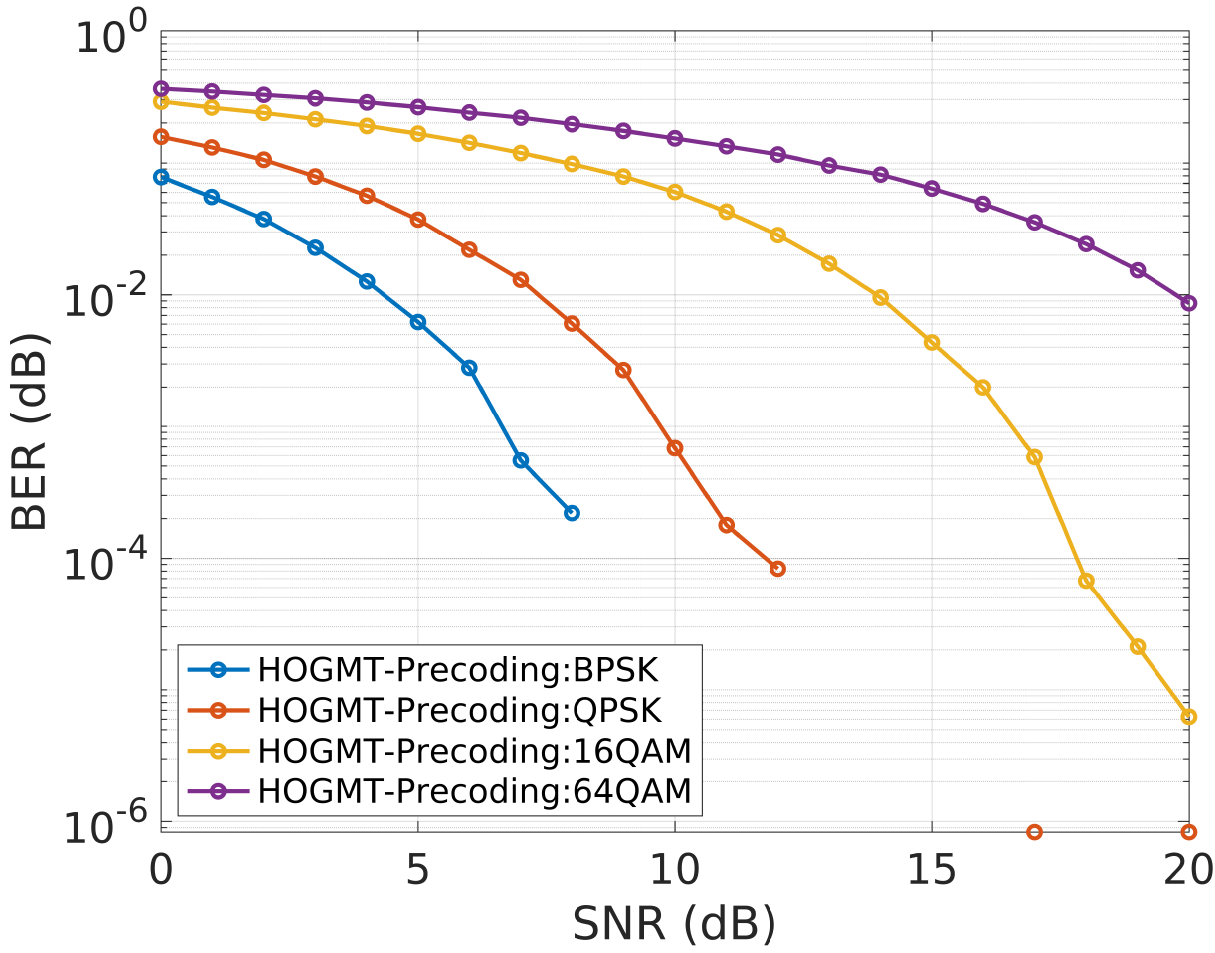}
  \caption{BER of HOGMT based spatial precoding for BPSK, QPSK, 16-QAM and 64-QAM modulations }
  \label{fig:Ber_space_multimode}
\end{subfigure}
  \caption{HOGMT based spatial precoding}
  \label{fig:spatial_precoding}
  \vspace{-10pt}
\end{figure*}
We observe that with low number of eigenfunctions ($\epsilon{=}10^{-1}$), the proposed precoding results in higher BER than DPC because precoding using these limited eigenfunctions is not sufficient to cancel all the spatial interference. 
However, with more eigenfunctions ($\epsilon{=}10^{-2}$), it achieves significantly lower BER compared to DPC for SNR ${>}9$dB and consequently only requires 3dB more SNR to achieve the same BER as the ideal case.
In contrast, while DPC is optimal in the sum rate at the transmitter is not optimal in terms of the BER at the receiver and its BER performance depends on the complementary decoding performed at the receiver.

\section{Conclusion}
\label{sec:conclusion}

In this work, we derived a high-order generalized version of Mercer's Theorem to decompose the general non-stationary channel kernel into 2-dimensional jointly orthogonal flat fading sub-channels (eigenfunctions).
Through theoretical analysis and simulations, we draw three firm conclusions for non-stationary channels: 1) The 2-dimensional eigenfunctions are sufficient to completely derive the second-order statistics of the non-stationary channel and consequently leads to an 
unified characterization of any wireless channel, 2) precoding by combining these eigenfunctions with optimally derived coefficients 
mitigates the spatio-temporal interference, 
and 3) the precoded symbols when propagated over the non-stationary channel directly reconstruct the data symbols at the receiver when combined with calculated coefficients, consequently alleviating the need for complex complementary decoding at the receiver.
Therefore, the encouraging results from this work will form the core of robust and unifed characterization and highly reliable communication over nonstationary channels, supporting emerging application.
\appendices
\label{appendix:precoding}







\section{Related work}
\label{App:related_NLP}


We categorize the related work into three categories:

\noindent
\textbf{Characterization of Non-Stationary Channels:}
Wireless channel characterization in the literature typically require several local and global (in space-time dimensions) higher order statistics to characterize or model non-stationary channels, due to their time-varying statistics. 
These statistics cannot completely characterize the non-stationary channel, however are useful in reporting certain properties that are required for the application of interest such as channel modeling, assessing the degree of stationarity etc.
Contrarily, we leverage the 2-dimensional eigenfunctions that are decomposed from the most generic representation of any wireless channel as a spatio-temporal channel kernel.
These spatio-temporal eigenfunctions can be used to extract any higher order statistics of the channel as demonstrated in Section \rom{3}, and hence serves as a complete characterization of the channel.
Furthermore, since this characterization can also generalize to stationary channels, it is a unified characterization for any wireless channel.
Beyond characterizing the channel, these eigenfunctions are the core of the precoding algorithm.

\noindent
\textbf{Precoding Non-Stationary Channels:}
Although precoding non-stationary channels is unprecedented in the literature \cite{AliNS0219}, we list the most related literature for completeness. 
The challenge in precoding non-stationary channels is the lack of accurate models of the channel and the (occasional) CSI feedback does not fully characterize the non-stationarities in its statistics. This leads to suboptimal performance using state-of-the-art precoding techniques like Dirty Paper Coding which assume that complete and accurate knowledge of the channel is available, while the CSI is often outdated in non-stationary channels. 
While recent literature present  attempt to deal with imperfect CSI by modeling the error in the CSI \cite{HatakawaNLP2012, HasegawaTHP2018, GuoTHP2020, DietrichTHP2007, CastanheiraPGS2013, WangTHP2012, MazroueiTDVP2016, Jacobsson1DAC2017}, they are limited by the assumption the channel or error statistics are stationary or WSSUS at best.
Another class of literature, attempt to deal with the impact of outdated CSI~\cite{AndersonLP2008,Zeng2012LP} in time-varying channels by quantifying this loss or relying statistical CSI. These methods are not directly suitable for non-stationary channels, as the time dependence of the statistics may render the CSI (or its statistics) stale, consequently resulting in precoding error. 

\noindent
\textbf{Space-Temporal Precoding:}
While, precoding has garnered significant research, spatio-temporal interference is typically treated as two separate problems, where spatial precoding at the transmitter aims to cancel inter-user and inter-antenna interference, while equalization at the receiver mitigates inter-carrier and inter-symbol interference.
Alternately, \cite{hadani2018OTFS} proposes to modulate the symbols such that it reduces the cross-symbol interference in the delay-Doppler domain, but requires equalization at the receiver to completely cancel such interference in practical systems. 
Moreover, this approach cannot completely minimize the joint spatio-temporal interference that occurs in non-stationary channels since their statistics depend on the time-frequency domain in addition to the delay-Dopper domain (explained in Section \rom{2}).
While spatio-temporal block coding techniques are studied in the literature \cite{Cho2010MIMObook} they add redundancy and hence incur a communication overhead to mitigate interference, which we avoid by precoding. 
These techniques are capable of independently canceling the interference in each domain, however are incapable of mitigating interference that occurs in the joint spatio-temporal domain in non-stationary channels. 
We design a joint spatio-temporal precoding that leverages the extracted 2-D eigenfunctions from non-stationary channels to mitigate interference that occurs on the joint space-time dimensions, which to the best of our knowledge is unprecedented in the literature.

\section{Proofs on Unified Characterization}
\label{app:characterization}

\subsection{Proof of Lemma 1: Generalized Mercer's Theorem}
\label{App:gmt}

\begin{proof}
Consider a 2-dimensional process $K(t,t') \in L^2(Y \times X)$, where $Y(t)$ and $X(t')$ are square-integrable zero-mean random processes with covariance function $K_{Y}$ and $K_{X}$, respecly. 
The projection of $K(t, t')$ onto $X(t')$ is obtained as in \eqref{eq:projection},
\begin{align}
\label{eq:projection}
    & C(t) = \int K(t, t') X(t') ~dt'
\end{align}

Using \textit{Karhunen–Loève Transform} (KLT), $X(t')$ and $C(t)$ are both decomposed as in \eqref{eq:X_t} and \eqref{eq:C_t}, 
\begin{align}
    &X(t') = \sum_{i = 1}^{\infty} x_{i} \phi_{i}(t') \label{eq:X_t}\\
    &C(t) = \sum_{j = 1}^{\infty} c_{j} \psi_{j}(t) \label{eq:C_t}
\end{align}
where $x_i$ and $c_j$ are both random variables with $\mathbb{E}\{x_i x_{i'}\} {=} \lambda_{x_i} \delta_{ii'}$ and $\mathbb{E}\{c_j c_{j'}\} {=} \lambda_{c_j} \delta_{jj'}$.  $\{\lambda_{x_i}\}$, $\{\lambda_{x_j}\}$ $\{\phi_i(t')\}$ and $\{\psi_j(t)\}$ are eigenvalues and eigenfuncions, respectively.
%

Let us denote $n{=}i{=}j$ and $\sigma_n {=} \frac{c_n}{x_n}$, and assume that $K(t,t')$ can be expressed as in \eqref{eq:thm_K_t},
\begin{align}
\label{eq:thm_K_t}
    K(t,t') = \sum_n^\infty \sigma_n \psi_{n}(t) \phi_{n}(t')
\end{align}
We show that \eqref{eq:thm_K_t} is a correct representation of $K(t,t')$ by proving \eqref{eq:projection} holds under this definition. 
We observe that by substituting \eqref{eq:X_t} and \eqref{eq:thm_K_t} into the right hand side of \eqref{eq:projection} we have that,
\begin{align}
    & \int K(t, t') X(t') ~dt' \nonumber \\
    & = \int \sum_n^\infty \sigma_n \psi_{n}(t) \phi_{n}(t') \sum_{n}^{\infty} x_{n} \phi_{n}(t') ~dt' \nonumber \\
    & = \int \sum_n^\infty \sigma_n x_n \psi_n(t) |\phi_n(t')|^2 \nonumber\\
    & + \sum_{n'\neq n}^ \infty \sigma_{n} x_{n'} \psi_{n}(t) \phi_{n}(t') \phi_{n'}^*(t') ~d t' \nonumber \\
    & = \sum_n^\infty c_n \psi_n(t) = C(t)
\end{align}
which is equal to the left hand side of \eqref{eq:projection}. 
Therefore, \eqref{eq:thm_K_t} is a correct representation of $K(t,t')$.

\end{proof}

\subsection{Proof of Theorem 1: High Order Generalized Mercer's Theorem (HOGMT}
\label{app:hogmt}

\begin{proof}
Given a 2-D process $X(\gamma_1, \gamma_2)$, the eigen-decomposition using Lemma 1 is given by,
\begin{equation}
\label{eq:thm1_1}
    X(\gamma_1, \gamma_2) = \sum_{n}^{\infty} x_{n} e_n(\gamma_1) s_n(\gamma_2)
\end{equation}

Letting $\psi_n(\gamma_1,\gamma_2) {=} e_n(\gamma_1) s_n(\gamma_2)$, and substituting it in \eqref{eq:thm1_1} we have that,

\begin{equation}
\label{eq:2d_klt}
    X(\gamma_1, \gamma_2) = \sum_{n}^{\infty} x_{n} \phi_n(\gamma_1,\gamma_2)
\end{equation}
where $\phi_n(\gamma_1,\gamma_2)$ are 2-D eigenfunctions with the property \eqref{eq:prop1}.
\begin{equation}
\label{eq:prop1}
\iint \phi_n(\gamma_1,\gamma_2) \phi_{n'}(\gamma_1,\gamma_2) ~d\gamma_1 ~d\gamma_2 = \delta_{nn'} 
\end{equation}

We observe that \eqref{eq:2d_klt} is the 2-D form of KLT. With iterations of the above steps, we obtain \textit{Higher-Order KLT} for $X(\gamma_1,\cdots,\gamma_Q)$ and $C(\zeta_1,\cdots,\zeta_P)$ as given by,
\begin{align}
   & X(\gamma_1,\cdots,\gamma_Q) = \sum_{n}^{\infty} x_{n} \phi_n(\gamma_1,\cdots,\gamma_Q) \\
   & C(\zeta_1,\cdots,\zeta_P) = \sum_{n}^{\infty} c_{n} \psi_n(\zeta_1,\cdots,\zeta_P)
\end{align}
where $C(\zeta_1,\cdots,\zeta_P)$ is the projection of $X(\gamma_1,\cdots,\gamma_Q)$ onto $K(\zeta_1,\cdots,\zeta_P; \gamma_1,\cdots, \gamma_Q)$.

Then following similar steps as in Appendix~\ref{App:gmt} we get \eqref{eq:col}. 
\begin{align}
\label{eq:col}
& K(\zeta_1,\cdots,\zeta_P; \gamma_1,\cdots, \gamma_Q) \nonumber \\
& = \sum_{n}^ \infty \sigma_n \psi_n(\zeta_1,\cdots,\zeta_P) \phi_n(\gamma_1,\cdots, \gamma_Q)
\end{align}
\end{proof}

\section{Proofs on Eigenfunction based Precoding}
\label{app:precoding}

\subsection{Proof of Lemma 2}
\label{app:lem1}

\begin{proof}
Using 2-D KLT as in (13), $x(u,t)$ is expressed as,
\begin{equation}
    x(u,t) = \sum_{n}^ \infty x_n \phi_n(u,t)
\end{equation}
where $x_n$ is a random variable with $E\{x_n x_{n'}\}{=} \lambda_n \sigma_{nn'} $ and $\phi_n(u,t)$ is a 2-D eigenfunction. 

Then the projection of $k_H(u,t;u',t')$ onto $\phi_n(u',t')$ is denoted by $ c_n(u,t)$ and is given by,
\begin{equation}
    c_n(u,t) =  \iint k_H(u,t;u',t') \phi_n(u',t') ~du' ~dt'
\end{equation}

Using the above, (28) is expressed as,
\begin{align}
\label{eq:obj_trans}
     & ||s(u,t) - Hx(u,t)||^2 = ||s(u,t) - \sum_n^ \infty x_n c_n(u,t)||^2
\end{align}

Let $\epsilon (x) {=} ||s(u,t) - \sum_n^ \infty x_n \phi_n(u,t)||^2$. Then its expansion is given by,
\begin{align}
\label{eq:ep}
     & \epsilon (x) = \langle s(u,t),s(u,t) \rangle - 2\sum_n ^ \infty  x_n \langle c_n(u,t),s(u,t) \rangle \\
     & + \sum_n^ \infty x_n^2 \langle c_n(u,t), c_n(u,t) \rangle \nonumber + \sum_n^ \infty \sum_{n' \neq n}^ \infty x_n x_{n'}  \langle c_n(u,t), c_{n'}(u,t) \rangle
\end{align}

Then the solution to achieve minimal $\epsilon(x)$ is obtained by solving for $\pdv{\epsilon(x)}{x_n} = 0$ as in \eqref{eq:solution}.
\begin{align}
\label{eq:solution}
    x_n^{opt} & {=}  \frac{\langle s(u,t), c_n(u,t) \rangle + \sum_{n'\neq n}^ \infty x_{n'} \langle c_{n'}(u,t), c_n(u,t) \rangle }{\langle c_n(u,t), c_n(u,t) \rangle}
\end{align}
where $\langle a(u,t), b(u,t) \rangle {=} \iint a(u,t) b^*(u,t) ~du ~dt$ denotes the inner product. 
Let $\langle c_{n'}(u,t), c_n(u,t) \rangle = 0$, i.e., the projections $\{ c_n(u,t)\}_n$ are orthogonal basis. Then we have a closed form expression for $x^{opt}$ as in \eqref{eq:x_opt}.
\begin{align}
\label{eq:x_opt}
    x_n^{opt} & {=}  \frac{\langle s(u,t), c_n(u,t) \rangle}{\langle c_n(u,t), c_n(u,t) \rangle}
\end{align}

Substitute \eqref{eq:x_opt} in \eqref{eq:ep}, it is straightforward to show that $\epsilon(x){=} 0$.
\end{proof}

\subsection{Proof of Theorem 2: Eigenfunction Precoding}
\label{app:thm_2}
\begin{proof}
The 4-D kernel $k_H(u,t;u',t')$ is decomposed into two separate sets of eigenfunction $\{\phi_n(u',t')\}$ and $\{\psi_n(u, t) \}$ using Theorem 1 as in (30). By transmitting the conjugate of the eigenfunctions, $\phi_n(u',t')$ through the channel $H$, we have that,  
\begin{align}
   & H \phi_n^*(u',t') = \iint k_H(u,t;u',t') \phi_n^*(u',t') ~du' ~d t' \nonumber \\ 
   & {=} \iint \sum_{n}^ \infty \{\sigma_n \psi_n(u,t) \phi_n(u',t')\} \phi_n^*(u',t') ~d t' ~d f' \nonumber \\ 
   & {=} \iint \sigma_n \psi_n(u,t) |\phi_n(u',t')|^2 \nonumber\\
   & + \sum_{n'\neq n}^ \infty \sigma_{n'} \psi_{n'}(u,t) \phi_{n'}(u',t')\ \phi_n^*(u',t') ~du' ~d t' \nonumber \\
   & {=} \sigma_n \psi_n(u,t)
\end{align}
where $\psi_n(u,t)$ is also a 2-D eigenfunction with the orthogonal property as in (31). 

From Lemma 2, if the set of projections, $\{c_n(u,t)\}$ is the set of eigenfunctions, $\{\psi_n(u,t)\}$, which has the above orthogonal property, we achieve the optimal solution as in \eqref{eq:x_opt}. Therefore, let $x(u,t)$ be the linear combination of $\{\phi_n^*(u,t)\}$ with coefficients $\{x_n\}$ as in \eqref{eq:construct},

\begin{equation}
\label{eq:construct}
    x(u,t) = \sum_n^ \infty x_n \phi_n^*(u,t) 
\end{equation}

Then \eqref{eq:obj_trans} is rewritten as in \eqref{eq:obj_trans2},
\begin{align}
\label{eq:obj_trans2}
     & ||s(u,t) - Hx(u,t)||^2 = ||s(u,t) - \sum_n^ \infty x_n \sigma_n \psi(u,t)||^2
\end{align}
   
Therefore, optimal $x_n$ in \eqref{eq:x_opt} is obtained as in \eqref{eq:opt},

\begin{equation}
\label{eq:opt}
    x_n^{opt} = \frac{\langle s(u,t), \psi_n(u,t) \rangle}{\sigma_n} 
\end{equation}

Substituting \eqref{eq:opt} in \eqref{eq:construct}, the transmit signal is given by \eqref{eq:x_opt2},
\begin{equation}
\label{eq:x_opt2}
    x(u,t) = \sum_n^ \infty \frac{\langle s(u,t), \psi_n(u,t) \rangle}{\sigma_n} \phi_n^*(u,t). 
\end{equation}
\end{proof}

\subsection{Proof of Corollary 1}
\label{app:EP_space}
\begin{proof}
First we substitute the 4-D kernel $k_H(u,t;u',t')$ with the 2-D kernel $k_H(u,u')$ in Theorem 2 which is then decomposed by the 2-D HOGMT. Then following similar steps as in Appendix~\ref{app:thm_2} it is straightforward to show (34).
\end{proof}

\section{Results on Interference}
\label{App:results_interference}
\begin{figure}[h]
  \centering
  \includegraphics[width=1\linewidth]{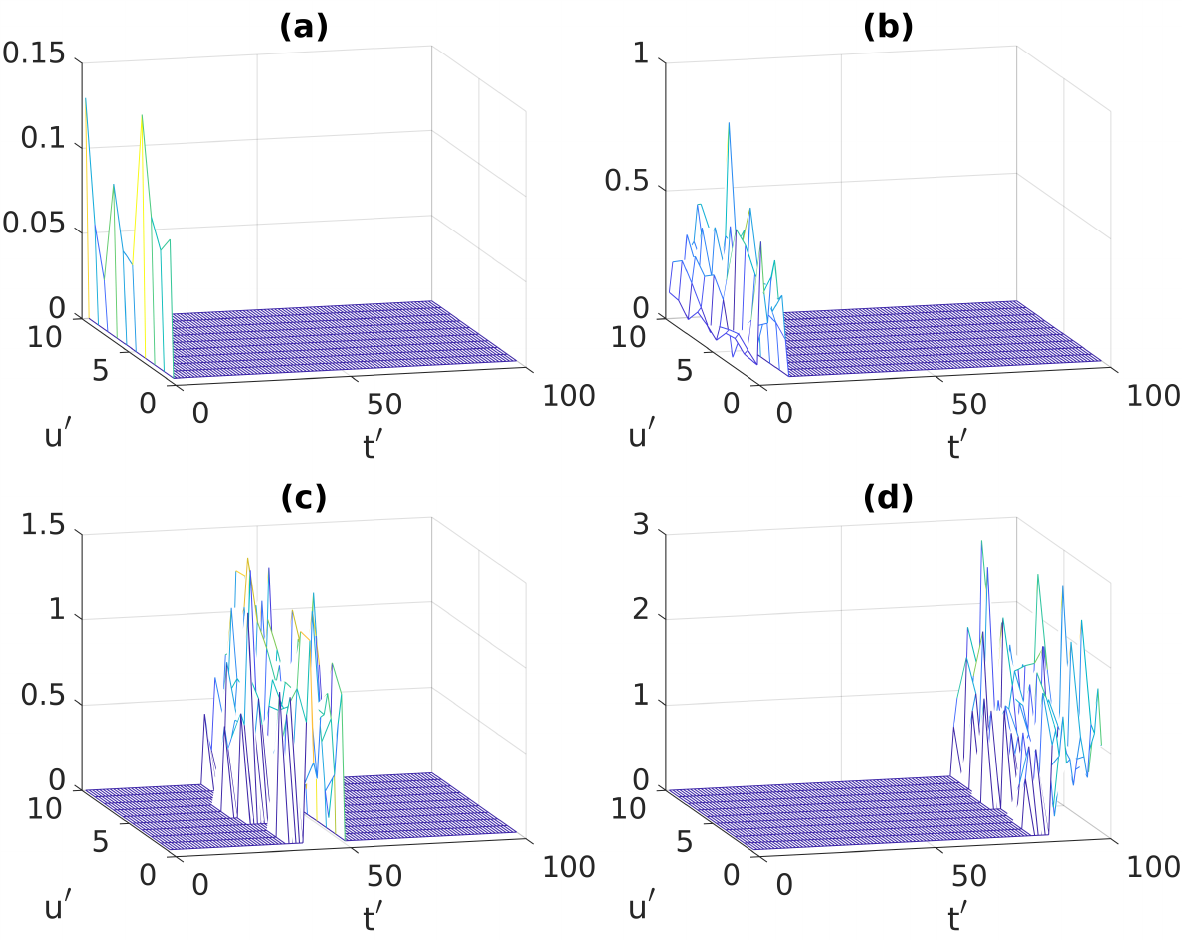}
  \caption{Kernel $k_H(u,t;u',t')$ for $u {=} 1$ at a) $t {=} 1$, b) $t {=} 10 $, c) $t {=} 50$ and d) $t {=} 100$.}
  \label{fig:hst_1_10_50_100}
\label{fig:hst_1_10_50_100}
\end{figure}
Figure~\ref{fig:hst_1_10_50_100} shows the channel response for user $u {=} 1$ at $t{=}1$, $t{=}10$, $t{=}50$ and $t{=}100$, where at each instance, the response for user $u {=} 1$ is not only affected by its own delay and other users' spatial interference, but also affected by other users' delayed symbols. 
This is the cause of joint space-time interference which necessitates joint precoding in the 2-dimensional space using eigenfunctions that are jointly orthogonal.

\bibliographystyle{IEEEtran}
\bibliography{references}

\end{document}